\documentclass[pra,onecolumn,groupedaddress,superscriptaddress,amsmath,amssymb,a4paper,longbibliography]{revtex4-2}
\usepackage{geometry}
\usepackage{graphicx}
\usepackage{ragged2e}
\usepackage{mathrsfs}
\usepackage{amssymb}
\usepackage{amsmath}
\usepackage{amsthm}
\usepackage{amsbsy}
\usepackage{bm}
\usepackage[T1]{fontenc}
\usepackage{color}
\newtheorem{proposition}{Proposition}

\theoremstyle{definition}

\newtheorem{remark}{Remark}

\newcommand{\bra}[1]{\langle #1|}
\newcommand{\ket}[1]{| #1 \rangle }

\usepackage{comment}

\DeclareMathOperator{\Tr}{Tr}

\usepackage{xcolor}
\usepackage[normalem]{ulem}

\usepackage{url}
\usepackage{hyperref}
\usepackage[nameinlink,capitalize]{cleveref}
\hypersetup{
  colorlinks = true,
  linkcolor = purple,
  citecolor = purple,
  urlcolor = blue,
  linkbordercolor = white
}

\usepackage{orcidlink}

\begin{document}
\title{Sparse positive maps on qutrits with exact nondecomposability thresholds and PPT-entanglement transitions}

\date{\today}

\author{Davide Poderini\,\orcidlink{0000-0003-0577-1608}}
	\email[Davide Poderini: ]{davide.poderini@unipv.it}
	\affiliation{Dipartimento di Fisica, Università degli Studi di Pavia, Via Agostino Bassi 6, I-27100, Pavia, Italy}

\author{Angela Rosy Morgillo\,\orcidlink{0009-0006-6142-0692}}
	\email[Angela Rosy Morgillo: ]{angelarosy.morgillo01@ateneopv.it}
	\affiliation{Dipartimento di Fisica, Università degli Studi di Pavia, Via Agostino Bassi 6, I-27100, Pavia, Italy}
	\affiliation{INFN Sezione di Pavia, Via Agostino Bassi 6, I-27100, Pavia, Italy}

\author{Fabio Benatti\,\orcidlink{0000-0002-0712-2057}}
	\email[Fabio Benatti: ]{fbenatti@units.it}
	\affiliation{Dipartimento di Fisica, Università di Trieste, Strada Costiera 11, 34151 Trieste, Italy}
	\affiliation{INFN Sezione di Trieste, via Alfonso Valerio 2, I-34151, Trieste, Italy}

\author{Fabio Anselmi\,\orcidlink{0000-0002-0264-4761}}
	\email[Fabio Anselmi: ]{fabio.anselmi@units.it}
	\affiliation{Dipartimento di Matematica, Informatica e Geoscienze, Università di Trieste, via Alfonso Valerio 2, 34127 Trieste, Italy}
	\affiliation{MIT, 77 Massachusetts Ave, Cambridge, MA 02139, United States of America}

\author{Chiara Macchiavello\,\orcidlink{0000-0002-2955-8759}}
	\email[Chiara Macchiavello: ]{chiara.macchiavello@unipv.it}
	\affiliation{Dipartimento di Fisica, Università degli Studi di Pavia, Via Agostino Bassi 6, I-27100, Pavia, Italy}
	\affiliation{INFN Sezione di Pavia, Via Agostino Bassi 6, I-27100, Pavia, Italy}

\author{Massimiliano F. Sacchi\,\orcidlink{0000-0002-8909-2196}}
	\email[Massimiliano F. Sacchi: ]{massimiliano.sacchi@unipv.it}
	\affiliation{CNR-Istituto di Fotonica e Nanotecnologie, Piazza Leonardo da Vinci 32, I-20133, Milano, Italy}
	\affiliation{Dipartimento di Fisica, Università degli Studi di
          Pavia, Via Agostino Bassi 6, I-27100, Pavia, Italy}

        \begin{abstract}
We study a family of sparse positive maps on qutrits for which
positivity, decomposability, and PPT entanglement can all be analysed
explicitly.  The block structure of the associated Choi matrices
reduces positivity to a Hermitian biquadratic form and leads to exact
positivity boundaries for three representative parametric families. For
the same families we determine the exact transition between
decomposable and non-decomposable maps and construct associated PPT
states of two classes. The first consists of witness-adapted
deformations naturally tied to the non-decomposability analysis. The
second consists of analytically tractable families whose full
PPT-entangled branch is detected by fixed positive maps, yielding
exact thresholds between separability and bound entanglement. For the
trace-preserving subclass, we further compare positivity with a recent
eigenvalue bound for $2$-positive maps, thereby making the gap between
positivity and higher-order positivity fully explicit within this
family.
\end{abstract}

\maketitle

\section{Introduction}
\label{sec:intro}

Entanglement is one of the defining features of quantum theory and a key
resource for quantum information processing, communication, sensing, and
simulation~\cite{horodecki2009quantum}. Determining whether a bipartite mixed
state is separable or entangled is, however, a notoriously difficult
problem. Among the most powerful approaches to this problem is the theory of
positive maps: a bipartite state is separable if and only if it remains
positive under $\mathrm{id}\otimes\Phi$ for every positive map
$\Phi$~\cite{HORODECKI19961,majewski2001acharacterization}. In finite
dimensions this viewpoint is naturally expressed through the
Choi--Jamiołkowski correspondence~\cite{jamiolkowski1972linear,choi1975completely},
while completely positive maps admit Kraus--Stinespring
representations~\cite{stinespring1955positive}. Positive maps are therefore
central both to the mathematical structure of quantum operations and to the
operational detection of entanglement. The present work identifies a
tractable qutrit family in which this map-based perspective can be developed
fully explicitly: positivity, decomposability, and the associated
PPT-entangled states can all be characterized in closed form.

Among positive maps, the non-decomposable ones are especially important.
Decomposable maps cannot witness PPT entangled states, i.e.\ states that are
entangled despite having positive partial transpose and therefore cannot be detected by the Peres-Horodecki criterion~\cite{peres1996separability,HORODECKI19961,benatti2004non}. 
For this reason, explicit families of positive maps for which positivity,
decomposability, and the associated PPT-entangled states can all be analyzed
in detail remain valuable tools for understanding the geometry of quantum
states and the structure of entanglement witnesses~\cite{terhal2001family,benatti2004quantum,bhattacharya2021generating}. 
In the family studied here, the sparsity of the Choi matrix makes it possible to obtain exact positivity boundaries, exact decomposability thresholds, and explicit PPT-state families whose separable and PPT-entangled regimes can be resolved analytically.

Despite substantial progress, analytically tractable families of non-decomposable positive maps are still comparatively scarce. 
Classical examples include the Choi map and its generalizations~\cite{cho1992generalized,jafarizadeh2006generalized,chruscinski2018generalizing,ha2013exposedness}, while more recent constructions rely on convex-geometric, algebraic, or structural techniques~\cite{kossakowski2003class,chruscinski2009geometry,hou2010characterization,marciniak2017merging,zwolak2013new,muller2018decomposability,majewski2017origin,jannesary2025class,gulati2025positive}. 
What remains relatively uncommon is a family simple enough to admit exact description of the boundaries in the space of positive maps and explicit constructions of the corresponding PPT-entangled state.
The maps studied here originally emerged from an optimization-based search for non-decomposable positive maps~\cite{morgillo2026maps}.
The main feature of the present family is the sparsity of its Choi matrix that makes it possible to carry out the analysis almost entirely in closed form.
In particular, positivity reduces to a Hermitian biquadratic form, complete positivity becomes a simple set of block constraints, and decomposability can be studied by pairing the maps with explicit PPT states.

Our main result is that this qutrit family provides an analytically controlled setting in which positivity, decomposability, and PPT entanglement can be studied on equal footing. 
For three representative parametric families we derive exact positivity regions and exact decomposability thresholds. 
We then construct associated PPT-state families that make the witnessing action of the maps explicit. Besides the witness-adapted deformations naturally tied to the decomposability analysis, we identify further one-parameter PPT families whose entire PPT-entangled branch is detected by fixed positive maps, thereby yielding exact separability thresholds. 
For the
trace-preserving subclass, we also relate our results to a recent spectral
bound for $2$-positive maps~\cite{boundoneigs2025}, which makes the gap
between positivity and higher-order positivity completely explicit within this
family.

The paper is organized as follows. In \cref{sec:theory} we review the
relevant notions on positive maps, decomposability, and Choi matrices. In
\cref{sec:map} we introduce the family under study, determine the spectrum of
its Choi matrix, and derive the complete-positivity conditions. In
\cref{sec:positivity} we obtain exact positivity regions for three
representative parametric families, while in \cref{sec:decomposability} we
determine the corresponding decomposability thresholds and, in the
trace-preserving case, compare them with the spectral bound
of~\cite{boundoneigs2025}. In \cref{sec:bound_ent_states} we construct the
associated PPT-state families, first in the witness-adapted form naturally
linked to the decomposability analysis and then in a sharper form whose full
PPT-entangled branch is detected by fixed witnesses. Finally,
\cref{sec:conclusion} summarizes the main conclusions and possible
directions for further work. Technical details are collected in the
appendices.
\section{Preliminaries and notation}
\label{sec:theory}

We briefly review the notions and tools on Hermiticity-preserving linear maps on matrix algebras that will be used throughout the paper, together with their relation to bipartite entanglement.

Let $\mathcal{B}(\mathbb{C}^d)$ denote the algebra of complex $d\times d$ matrices, and let $\Phi:\mathcal{B}(\mathbb{C}^d)\to\mathcal{B}(\mathbb{C}^d)$ be a linear map. We shall consider maps that preserve Hermiticity, namely maps such that $\Phi[X]^\dagger=\Phi[X]$ whenever $X^\dagger=X$.
Among such maps, the following classes play a central role:
\begin{enumerate}
\item
$\Phi$ is \emph{positive} if $\Phi[X]\succeq 0$ for every positive semidefinite operator $X\in\mathcal{B}(\mathbb{C}^d)$;

\item
$\Phi$ is \emph{$k$-positive} if the lifted map $\mathrm{id}_k\otimes\Phi$ is positive on $\mathcal{B}(\mathbb{C}^k)\otimes\mathcal{B}(\mathbb{C}^d)$, where $\mathrm{id}_k$ denotes the identity map on $\mathcal{B}(\mathbb{C}^k)$;

\item
$\Phi$ is \emph{completely positive} if it is $k$-positive for every $k\ge 1$; in particular, $1$-positivity coincides with ordinary positivity;

\item
$\Phi$ is \emph{completely co-positive} if $\Phi\circ T$ is completely positive, where $T$ denotes transposition on $\mathcal{B}(\mathbb{C}^d)$;

\item
$\Phi$ is \emph{decomposable} if it can be written as the sum of a completely positive map and a completely co-positive one.
\end{enumerate}

A basic tool in the analysis of such maps is the \emph{Choi matrix} associated with $\Phi$,
\begin{equation}
\label{eq:ChoiM}
C_\Phi=(\mathrm{id}_d\otimes\Phi)\ket{\psi}\bra{\psi}
=\sum_{n,m=1}^{d}\ket{n}\bra{m}\otimes\Phi[\ket{n}\bra{m}]
=
\begin{pmatrix}
\Phi[\ket{1}\bra{1}]&\cdots&\Phi[\ket{1}\bra{d}]\cr
\cdot&\cdots&\cdot\cr
\cdot&\cdots&\cdot\cr
\cdot&\cdots&\cdot\cr
\Phi[\ket{d}\bra{1}]&\cdots&\Phi[\ket{d}\bra{d}]
\end{pmatrix},
\end{equation}
where
\begin{equation}
\ket{\psi}=\sum_{n=1}^{d}\ket{n}\otimes\ket{n}\in \mathbb{C}^d\otimes\mathbb{C}^d
\label{eq:sym-vec}
\end{equation}
is the unnormalized maximally entangled vector in the computational basis. Choi's theorem states that $\Phi$ is completely positive if and only if $C_\Phi\succeq 0$~\cite{choi1975completely}. Equivalently, for maps on $\mathcal{B}(\mathbb{C}^d)$, $d$-positivity already implies complete positivity.

Whenever $\Phi$ is completely positive, it admits a Kraus--Stinespring representation of the form
\begin{equation}
\label{eq:KS}
\Phi[X]=\sum_{\alpha=1}^{d^2}\gamma_\alpha\,F^\dagger_\alpha\,X\,F_\alpha\,,
\end{equation}
where $\gamma_\alpha\geq 0$ are the eigenvalues of $C_\Phi$. The corresponding Kraus operators $F_\alpha$ are obtained from the eigenvectors $\ket{\phi_\alpha}$ of $C_\Phi$. Writing
\begin{equation}
\label{eq:kraus}
\ket{\phi_\alpha}=\sum_{n=1}^{d}\ket{n}\otimes\ket{\phi_{\alpha n}}
=(\mathrm{1}_d\otimes F^\dagger_\alpha)\ket{\psi}\,,
\end{equation}
their matrix elements are
\begin{equation}
\label{KSformula}
\langle m\vert F_\alpha^\dagger\vert n\rangle=\langle m\vert\phi_{\alpha n}\rangle\,.
\end{equation}

Ordinary positivity is weaker than complete positivity. A Hermiticity-preserving map $\Phi$ is positive if and only if its Choi matrix is \emph{block-positive}, namely if
\[
\bra{\chi}\otimes\bra{\phi}\,C_\Phi\,\ket{\chi}\otimes\ket{\phi}\geq 0
\qquad \forall\,\ket{\chi},\ket{\phi}\in\mathbb{C}^d.
\]
By convexity, this is equivalent to non-negativity on all separable states. The problem of characterizing positive but not completely positive maps can therefore be reformulated in spectral terms: how many eigenvalues of $C_\Phi$ may be negative, and how negative can they be, while block-positivity is still preserved?

This question is directly connected with entanglement detection. Indeed, separable states---that is, convex combinations of tensor-product states---remain positive under the action of $\mathrm{id}\otimes\Phi$ for every positive map $\Phi$. Conversely, a fundamental theorem states that a bipartite state is entangled if and only if it is detected by some positive but not completely positive map, meaning that $(\mathrm{id}\otimes\Phi)[\rho]$ fails to be positive for at least one such $\Phi$~\cite{HORODECKI19961}. In this sense, positive but not completely positive maps act as entanglement witnesses.

The prototypical example is the transposition map $T$. Its Choi matrix is the flip operator $V$, defined by
\[
V\ket{\chi}\otimes\ket{\phi}=\ket{\phi}\otimes\ket{\chi}.
\]
In dimension $d$, $V$ has eigenvalue $-1$ with degeneracy
$d(d-1)/2$. For bipartite systems of dimensions $2\times2$ and
$2\times3$, positivity under partial transposition is equivalent to
separability because in these cases all positive maps are
decomposable~\cite{woronowicz1976positive}. In higher dimensions this
equivalence breaks down: there exist entangled states that remain
positive under partial transposition. These are the so-called
\emph{bound-entangled states}, namely entangled states from which no
pure entanglement can be distilled.

For these reasons, the construction of new classes of positive but not completely positive maps remains a problem of primary interest.

\section{Qutrit map family and Choi structure}
\label{sec:map}

We now introduce the family of maps that will be studied throughout the paper. Let
$\Phi:\mathcal{B}(\mathbb{C}^3)\to \mathcal{B}(\mathbb{C}^{3})
$ act on a generic matrix $X\in\mathcal{B}(\mathbb{C}^3)$ as
\begin{equation}
\Phi[X]= \begin{pmatrix}
a X_{11}+c X_{22}+b X_{33}
&
w^*\, X_{21}
&
z\, X_{13}
\\[4pt]
w\, X_{12}
&
b\,X_{11}+a\,X_{22}+c\,X_{33}
&
0
\\[4pt]
z^*\,X_{31}
&
0
&
c\,X_{11}+b\,X_{22}+a\,X_{33}
\end{pmatrix},
\label{eq:map}
\end{equation}
where $a,b,c$ are non-negative real parameters and $w,z$ are complex numbers.

The map is proportional to a trace-preserving unital map. Indeed,
\[
\Tr[\Phi(X)]=(a+b+c)\Tr[X] \ ,
\qquad
\Phi(\mathrm{1}_3)=(a+b+c)\mathrm{1}_3 \ .
\]

The corresponding Choi matrix defined in \cref{eq:ChoiM} is
\begin{equation}
\setlength{\arraycolsep}{6pt}
C_\Phi=(\mathrm{id}_3\otimes\Phi)\ket{\psi}\bra{\psi}
=\left(\begin{array}{ccc|ccc|ccc}
 a & \cdot & \cdot & \cdot & \cdot & \cdot & \cdot & \cdot & z\\
 \cdot & b & \cdot & w & \cdot & \cdot & \cdot & \cdot & \cdot\\
 \cdot & \cdot & c & \cdot & \cdot & \cdot & \cdot & \cdot & \cdot\\ \hline
 \cdot & w^* & \cdot & c & \cdot & \cdot & \cdot & \cdot & \cdot\\
 \cdot & \cdot & \cdot & \cdot & a & \cdot & \cdot & \cdot & \cdot\\
 \cdot & \cdot & \cdot & \cdot & \cdot & b & \cdot & \cdot & \cdot\\ \hline
 \cdot & \cdot & \cdot & \cdot & \cdot & \cdot & b & \cdot & \cdot\\
 \cdot & \cdot & \cdot & \cdot & \cdot & \cdot & \cdot & c & \cdot\\
 z^* & \cdot & \cdot & \cdot & \cdot & \cdot & \cdot & \cdot & a
\end{array}\right) \ ,
\label{eq:ch1}
\end{equation}
where zero entries have been replaced by dots for readability.

Because $C_\Phi$ consists of two non-trivial $2\times2$
blocks together with five one-dimensional blocks, its spectrum can be
easily read off.  The eigenvalues are
\begin{eqnarray}
\label{eig1234567}
\gamma_{1,2}=a\pm|z|
\ ,\ \gamma_3=\gamma_6=b\ ,\ \gamma_4=\gamma_7=c\ ,\ \gamma_5=a\ ,
\gamma_{8,9}=\frac{b+c\pm\sqrt{(b-c)^2+4|w|^2}}{2}\ .
\end{eqnarray}
Therefore, the Choi matrix is positive semidefinite, and hence $\Phi$ is completely positive, if and only if
\begin{equation}
\label{CPcond}
a\geq|z|\ ,\ b\geq 0\ ,\ c\geq 0\ ,\ bc\geq |w|^2\ .
\end{equation}
In \cref{app:KS} we derive the corresponding Kraus-Stinespring expansion~\eqref{eq:KS}, whose Kraus operators are
\begin{align}
\label{eq:Kraus1}
F_1^\dag &= \frac{1}{\sqrt{2}}
\begin{pmatrix}
1&0&0\\
0&0&0\\
0&0&{\rm e}^{-i\arg(z)}
\end{pmatrix} \ , \quad
F_2^\dag =
\frac{1}{\sqrt{2}}
\begin{pmatrix}
1&0&0\\
0&0&0\\
0&0&-{\rm e}^{-i\arg(z)}
\end{pmatrix} \ , \quad
F_3^\dag =
\begin{pmatrix}
0&0&1\\
0&0&0\\
0&0&0
\end{pmatrix} \ ,
\\[4pt]
\label{eq:Kraus2}
F_4^\dag &=
\begin{pmatrix}
0&0&0\\
0&0&1\\
0&0&0
\end{pmatrix} \ , \quad
F_5^\dag =
\begin{pmatrix}
0&0&0\\
0&1&0\\
0&0&0
\end{pmatrix} \ , \quad
F_6^\dag =
\begin{pmatrix}
0&0&0\\
0&0&0\\
0&1&0
\end{pmatrix} \ , \quad
F_7^\dag =
\begin{pmatrix}
0&0&0\\
0&0&0\\
1&0&0
\end{pmatrix} \ ,
\\[4pt]
\label{eq:Kraus3}
F_8^\dag &=
\frac{1}{\sqrt{|w|^2+(\gamma_8-b)^2}}
\begin{pmatrix}
0&(\gamma_8-b){\rm e}^{-i\arg(w)}&0\\
|w|&0&0\\
0&0&0
\end{pmatrix} \ , \quad
F_9^\dag =
\frac{1}{\sqrt{|w|^2+(\gamma_9-b)^2}}
\begin{pmatrix}
0& (\gamma_9-b){\rm e}^{-i\arg(w)} &0\\
|w|&0&0\\
0&0&0
\end{pmatrix} \ .
\end{align}

\begin{remark}
\label{rem1}
As discussed in the previous section, $3$-positivity implies
$2$-positivity. In the present setting, however, the sparsity and
block structure of $C_\Phi$ in \cref{eq:ch1} make the
complete-positivity conditions in \cref{CPcond} not only sufficient
but also necessary for the map to be $2$-positive (see \cref{app3}). As a consequence, within this family every
$2$-positive map is automatically completely positive and therefore
trivially decomposable, in agreement
with~\cite{sanpera2001schmidt,yang2016all}.
\end{remark}

\section{Exact positivity regions}
\label{sec:positivity}

We now characterize positivity. Since positivity of
the map $\Phi$ is equivalent to block-positivity of its Choi matrix
$C_\Phi$, the problem reduces to the analysis of the associated
Hermitian biquadratic form. As shown in \cref{app4}, this condition
can be written as the requirement that
\begin{align}
Q(x,y) &= (a |x_1|^2 + b |x_2|^2 + c |x_3|^2)\, |y_1|^2
 + 2 \Re \big(w x_1 x_2^* y_1 y_2 ^* \big)
 + (c |x_1|^2 + a |x_2|^2 + b |x_3|^2)\, |y_2|^2 \nonumber \\
&\quad + (b |x_1|^2 + c |x_2|^2 + a |x_3|^2)\, |y_3|^2
 + 2 \Re \big(z x_1 x_3 ^* y_1^* y_3\big) 
\label{eq:bifo}
\end{align}
be non-negative for all $x,y\in\mathbb{C}^3$.

Equivalently,
\[
Q(x,y)=y^\dag M(x)\,y \ ,
\]
with
\begin{equation}
M(x)=
\begin{pmatrix}
a |x_1|^2 + b |x_2|^2 + c |x_3|^2 & w^* x_1^* x_2 & z ^* x_1 x_3^*\\
w x_1 x_2^* & c |x_1|^2 + a |x_2|^2 + b |x_3|^2 & 0\\
z x_1^* x_3 & 0 & b |x_1|^2 + c |x_2|^2 + a |x_3|^2
\end{pmatrix}\;,
\end{equation}
so that $Q(x,y)\ge0$ for all $y$ if and only if $M(x)\succeq0$ for all $x\neq0$.

Applying Sylvester's criterion~\cite{horn2012matrix}, positivity is therefore equivalent to the non-negativity of all principal minors of $M(x)$. Introducing the simplex variables
\begin{equation}
u_i:=\frac{|x_i|^2}{\sum_j |x_j|^2} \ ,
\qquad u_i\ge0 \ ,
\qquad \sum_{i=1}^3 u_i=1 \ ,
\end{equation}
the conditions become
\begin{align}
 &\textbf{(i)}\quad M_{11}=a u_1+b u_2+c u_3\ge 0\ , &
 &\textbf{(ii)}\quad M_{22}=c u_1+a u_2+b u_3\ge 0\ , \nonumber\\
 &\textbf{(iii)}\quad M_{33}=b u_1+c u_2+a u_3\ge 0\ , &
 &\textbf{(iv)}\quad \det M = M_{22}M_{33}S(u_1,u_2)\ge 0\ , \nonumber\\
 &\textbf{(v)}\quad M_{11}M_{22}-|w|^2u_1u_2\ge 0\ , &
 &\textbf{(vi)}\quad M_{11}M_{33}-|z|^2u_1u_3\ge 0\ ,
 \label{eq:eqspos}
\end{align}
where
\begin{equation}
S(u_1,u_2):=M_{11}-\frac{|w|^2u_1u_2}{M_{22}}-\frac{|z|^2u_1(1-u_1-u_2)}{M_{33}}\ ,
\qquad
u_3:=1-u_1-u_2\ .
\label{eq:sui}
\end{equation}

In the following we assume $a,b,c\ge0$. Under this assumption, it is enough to require $S(u_1,u_2)\ge0$ throughout the simplex. Indeed,
\begin{equation}
M_{22}S(u_1,u_2)=\bigl(M_{11}M_{22}-|w|^2u_1u_2\bigr)-\frac{M_{22}|z|^2u_1u_3}{M_{33}}\ge0\ ,
\end{equation}
which implies condition~\textbf{(v)}; multiplying instead by $M_{33}$ shows that condition~\textbf{(vi)} follows as well. We therefore focus on the sign of $S(u_1,u_2)$.

To obtain a fully analytical treatment, we specialize to three distinguished parametric families.

\subsection{\texorpdfstring{Case 1: $a\geq 0$, $b=c\geq 0$ and $w=z=1$}{Case 1: $a\ge0$, $b=c\ge0$, $w=z=1$}}

\begin{proposition}
\label{prop:pos_case1}
Consider Case~1, namely $a\ge0$, $b=c\ge0$, and $w=z=1$. Then the map $\Phi$ is positive if and only if
\begin{equation}
b\ge b_{\min}(a) \ ,
\qquad
b_{\min}(a)=
\begin{cases}
1-a \,, & 0\le a\le \tfrac12 \,,\\[1mm]
1+\dfrac{a-\sqrt{9a^2+8a}}{4} \,, & \tfrac12\le a \le \sqrt2 \,,\\[2mm]
0 \,, & a\ge \sqrt2 \,.
\end{cases}
\label{eq:prop_bmin_case1}
\end{equation}
Equivalently, the positivity boundary is exactly the graph of $b_{\min}(a)$.
\end{proposition}

\begin{proof}
Under these assumptions, the potentially critical eigenvalues of the Choi matrix reduce to $\gamma_{1,2}=a\pm1$ and $\gamma_{8,9}=b\pm1$. Hence $C_\Phi$ is not positive semidefinite, and therefore $\Phi$ is not completely positive, whenever either $0\le a<1$ or $0\le b<1$.

The function~\eqref{eq:sui} becomes
\begin{equation}
S(u_1,u_2)=(a-b)u_1+b-u_1\left[\frac{u_2}{b+(a-b)u_2}+\frac{u_3}{b+(a-b)u_3}\right].
\label{eq:S_case1}
\end{equation}
The positivity boundary is obtained by finding the smallest value $b_{\min}(a)$ such that $S(u_1,u_2)\ge0$ throughout the simplex.

Fix $u_1=t\in(0,1]$ and regard $S(t,u_2)$ as a function of $u_2\in[0,1-t]$. Its first two derivatives are
\begin{align}
\frac{\partial S}{\partial u_2}
&=-b u_1 \left[ \frac{1}{\bigl(b+(a-b)u_2\bigr)^2}
- \frac{1}{\bigl(b+(a-b)(1-u_1-u_2)\bigr)^2} \right] \,,
\label{eq:du_case1}\\
\frac{\partial^2 S}{\partial u_2^2}
&= 2b(a-b)u_1 \left[ \frac{1}{\bigl(b+(a-b)u_2\bigr)^3}
+ \frac{1}{\bigl(b+(a-b)(1-u_1-u_2)\bigr)^3} \right].
\label{eq:d2u_case1}
\end{align}
Thus, for fixed $t$, the minimizer is the symmetric point $u_2=u_3=(1-t)/2$ when $a\ge b>0$, whereas for $0<a<b$ the minimum is attained at a boundary point, say $u_2=1-t$.

\subparagraph*{\textbf{Subcase $\bm{a\ge b>0}$.}}
Evaluating at the symmetric point gives
\begin{equation}
S\!\left(t,\frac{1-t}{2}\right)=\frac{N(t)}{D(t)}\,,
\qquad t\in[0,1]\,,
\label{eq:S_case1_symm}
\end{equation}
where
\begin{align}
N(t) &= [2-(a-b)^2]t^2 + (a^2-ab-2)t + b(a+b)\,,
\label{eq:N_case1}\\
D(t) &= (b-a)t + a + b\,.
\label{eq:D_case1}
\end{align}
Since $D(t)\ge D(1)=2b>0$, the sign of $S$ is determined by the quadratic numerator $N(t)$. Moreover, $N(0)=b(a+b)>0$ and $N(1)=2ab>0$.

If $a-b\ge\sqrt{2}$, then $N$ is concave, and positivity at the endpoints implies $N(t)>0$ for all $t\in[0,1]$. If instead $0\le a-b<\sqrt{2}$, then $N$ is convex, and positivity on the whole interval holds provided either the vertex lies outside $[0,1]$ or inside with $N(t^*)\geq 0$.  These conditions reduce to
\begin{equation}
b>1+\frac{a-\sqrt{9a^2+8a}}{4}\,,
\qquad 0\le a-b<\sqrt{2}\,.
\label{eq:bound_ab_new}
\end{equation}

\subparagraph*{\textbf{Subcase $\bm{b>a>0}$.}}
At the boundary one has
\begin{equation}
S(t,1-t)=\frac{N(t)}{b+(a-b)(1-t)}\,,
\qquad t\in[0,1]\,,
\label{eq:S_case1_bdry}
\end{equation}
with
\begin{equation}
N(t)=[1-(b-a)^2](t^2-t)+ab\,.
\label{eq:N_case1_bdry}
\end{equation}
The denominator is strictly positive on $[0,1]$, so positivity is again determined by $N(t)$.

If $b-a\ge1$, then $N(t)\ge ab>0$. If $0<b-a<1$, then $N$ is convex and its minimum is attained at $t=\tfrac12$, where
\begin{equation}
N\!\left(\tfrac12\right)=\frac{(a+b)^2-1}{4}\,.
\end{equation}
Therefore positivity is equivalent to $b>1-a$.

Combining the two regimes, the positivity boundary is the piecewise function
\begin{equation}
b_{\min}(a)=
\begin{cases}
1-a \,, & 0\le a\le \tfrac12 \,,\\[1mm]
1+\dfrac{a-\sqrt{9a^2+8a}}{4} \,, & \tfrac12\le a \le \sqrt2 \,,\\[2mm]
0 \,, & a\ge \sqrt2 \,,
\end{cases}
\label{eq:bmin_case1}
\end{equation}
and the map is positive if and only if $b\ge b_{\min}(a)$.

The complete-positivity region is defined by $a\ge1$ and $b\ge1$; see \cref{fig:positivity_region} for the resulting geometry.
\end{proof}

\begin{figure}[hbt]
 \includegraphics[width=0.4\textwidth]{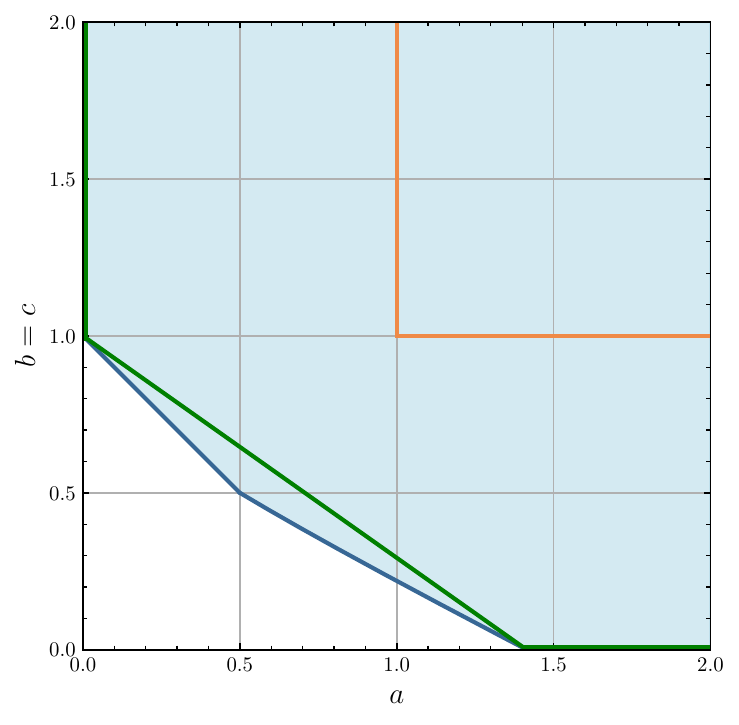}
 \caption{Boundaries for different properties of the map in
   \cref{eq:map}, for the case $b=c\ge 0$ and $w=z=1$. The blue-shaded
   region above the blue boundary corresponds to positivity [see
     \cref{eq:bmin_case1}]; the orange boundary delimits complete
   positivity; the area above the green segment corresponds to
   decomposability.}
 \label{fig:positivity_region}
\end{figure}

\begin{remark}
The Choi matrix of the transposition map has three negative eigenvalues equal to $-1$ in dimension $d=3$. By contrast, within the present class, the Choi matrices of positive but not completely positive maps have at most two negative eigenvalues.
\end{remark}

\subsection{\texorpdfstring{Case 2: $0\leq a\leq 1$, $b=c=(1-a)/2$ and $w=z\geq 0$}{Case 2: $0\le a\le 1$, $b=c=(1-a)/2$, $w=z\ge0$}}

\begin{proposition}
\label{prop:pos_case2}
Consider Case~2, namely $0\le a\le1$, $b=c=(1-a)/2$, and $w=z\ge0$. Then the map $\Phi$ is positive if and only if
\begin{equation}
0\le w\le w_{\max}(a) \ ,
\qquad
w_{\max}(a)=
\begin{cases}
\dfrac{1+a}{2} \,, & 0\le a<\tfrac13 \,,\\[2mm]
\dfrac{1-a+\sqrt{a+a^2}}{2} \,, & \tfrac13\le a\le1 \,.
\end{cases}
\label{eq:prop_wmax_case2}
\end{equation}
Equivalently, the positivity region is the region below the graph of $w_{\max}(a)$.
\end{proposition}

\begin{proof}
In this case the map is trace-preserving. The complete-positivity conditions reduce to
\begin{equation}
 w \le a \,, \quad \text{for } a \in [0,1/3]\,,
 \qquad
 w \le \frac{1-a}{2} \,, \quad \text{for } a \in [1/3,1]\,.
\end{equation}

The function~\eqref{eq:sui} becomes
\begin{equation}
S(u_1,u_2)=\frac{3a-1}{2}u_1+\frac{1-a}{2}-2w^2u_1\left[\frac{u_2}{1-a+(3a-1)u_2}+\frac{u_3}{1-a+(3a-1)u_3}\right].
\label{eq:S_case2}
\end{equation}
The positivity boundary is obtained by finding the largest value $w_{\max}(a)$ such that $S(u_1,u_2)\ge0$ throughout the simplex.

Fix $u_1=t\in(0,1]$ and regard $S(t,u_2)$ as a function of $u_2\in[0,1-t]$. Its first and second derivatives are
\begin{align}
\frac{\partial S}{\partial u_2}
&= -2w^2(1-a)u_1\left[\frac{1}{\bigl(1-a+(3a-1)u_2\bigr)^2}-\frac{1}{\bigl(1-a+(3a-1)(1-u_1-u_2)\bigr)^2}\right] \,,
\label{eq:du_case2}\\
\frac{\partial^2 S}{\partial u_2^2}
&= 4w^2(1-a)(3a-1)u_1\left[\frac{1}{\bigl(1-a+(3a-1)u_2\bigr)^3}+\frac{1}{\bigl(1-a+(3a-1)(1-u_1-u_2)\bigr)^3}\right].
\label{eq:d2u_case2}
\end{align}
Thus, for fixed $t$, the minimizing configuration is the symmetric point $u_2=u_3=(1-t)/2$ when $1/3\le a<1$, whereas the minimum is attained at the boundary $u_2=1-t$ for $0<a<1/3$.

\subparagraph*{\textbf{Subcase $\bm{1/3\le a<1}$.}}
At the symmetric point one finds
\begin{equation}
S\!\left(t,\frac{1-t}{2}\right)=\frac{N(t)}{D(t)} \,,
\qquad t\in[0,1] \,,
\label{eq:S_case2_symm}
\end{equation}
with
\begin{align}
N(t) &= [8w^2-(1-3a)^2 ]t^2 -2[4w^2+a(1-3a)]t +(1-a^2) \,,
\label{eq:N_case2}\\
D(t) &= 2[(1-3a)t+a+1] \,.
\label{eq:D_case2}
\end{align}
For $1/3\le a<1$ one has $D(t)>0$ for all $t\in[0,1]$, so the sign of $S$ is determined by the quadratic numerator.

If $w\le (3a-1)/(2\sqrt2)$, then $N$ is concave and positivity at the endpoints implies $N(t)>0$ on $[0,1]$. If $w>(3a-1)/(2\sqrt2)$, then $N$ is convex and positivity holds provided either the vertex lies outside $[0,1]$ or inside with $N(t^*)\geq 0$. This yields
\begin{equation}
w_{\max}(a)=\frac{1-a+\sqrt{a+a^2}}{2} \,,
\qquad \frac13\le a\le1 \,.
\label{eq:b1}
\end{equation}

\subparagraph*{\textbf{Subcase $\boldsymbol{0<a<1/3}$.}}
At the boundary one has
\begin{equation}
S(t,1-t)=\frac{N(t)}{D(t)} \,,
\qquad t\in[0,1] \,,
\label{eq:S_case2_bdry}
\end{equation}
with
\begin{align}
N(t) &= [4w^2 -(1-3a)^2](t^2-t)+2a(1-a) \,,
\label{eq:N_case2_bdry}\\
D(t) &= 2[(1-3a)t+2a] \,.
\label{eq:D_case2_bdry}
\end{align}
Again, $D(t)>0$ for all $t\in[0,1]$.

If $w<(1-3a)/2$, then $N(t)\ge2a(1-a)>0$. If $w\ge(1-3a)/2$, then $N$
is convex and its minimum is attained at $t=\tfrac12$, so positivity
is equivalent to
\begin{equation}
w_{\max}(a)=\frac{1+a}{2} \,,
\qquad 0\le a<\frac13 \,.
\label{eq:b2}
\end{equation}

In summary, the positivity boundary is the piecewise function
\begin{equation}
w_{\max}(a)=
\begin{cases}
\dfrac{1+a}{2} \,, & 0\le a<\tfrac13 \,,\\[2mm]
\dfrac{1-a+\sqrt{a+a^2}}{2} \,, & \tfrac13\le a\le1 \,,
\end{cases}
\label{eq:wmax_case2}
\end{equation}
and the map is positive if and only if $0\le w\le w_{\max}(a)$. The resulting geometry is shown in \cref{fig:boundary2}.
\end{proof}

\begin{figure}[htb]
\includegraphics[width=0.4\textwidth]{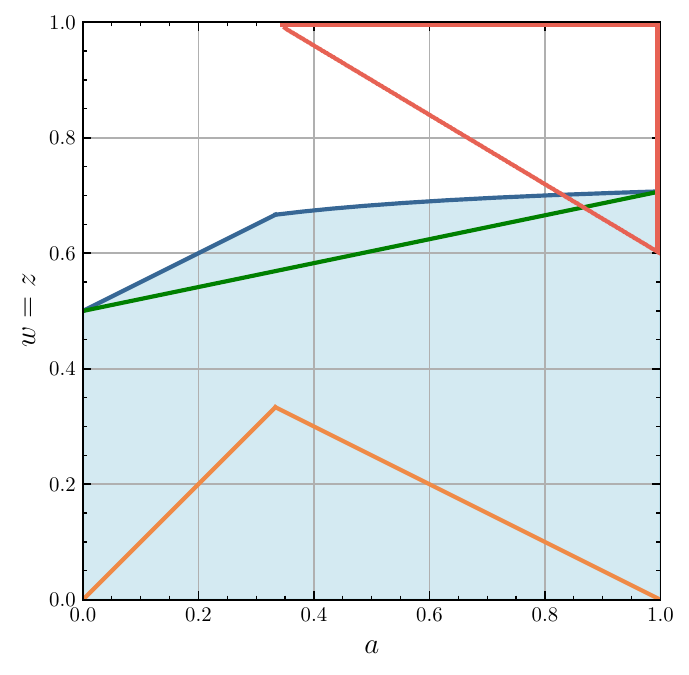}
\caption{Boundaries for different properties of the map in \cref{eq:map}, for the case $b=c=(1-a)/2$ and $w=z \geq 0$. The blue-shaded region below the blue roof indicates positivity [see \cref{eq:wmax_case2}]; the orange boundary delimits complete positivity; the area below the green segment corresponds to decomposability; the region inside the red triangle shows where the bound in \cref{eq:bb} is violated.}
\label{fig:boundary2}
\end{figure}

\begin{remark}
For $w\neq z$, the two branches forming the blue roof of the positivity region generally have different shapes, and their junction shifts toward larger values of $a$. Numerical examples are shown in \cref{fig:general_positivity}, where we define $p\equiv z/w$. The linear branch is given by
\begin{equation}
 w_{\max}(a)=\frac{1+a}{2\max\{1,p\}} \,,
 \label{eq:general_boundary}
\end{equation}
and at $a=1$ one finds $w_{\max}(1)=1/\sqrt{1+p^2}$.
\end{remark}

\begin{figure}[htb]
 \centering
 \begin{minipage}[b]{0.42\textwidth}
 \centering
 \includegraphics[width=\textwidth]{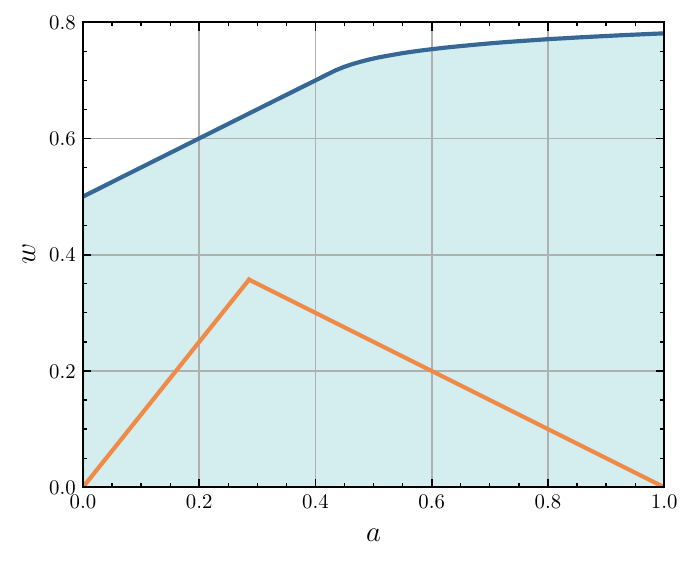}
 \\ \textbf{(a) $p=0.8$}
 \end{minipage}
 \hfill
 \begin{minipage}[b]{0.42\textwidth}
 \centering
 \includegraphics[width=\textwidth]{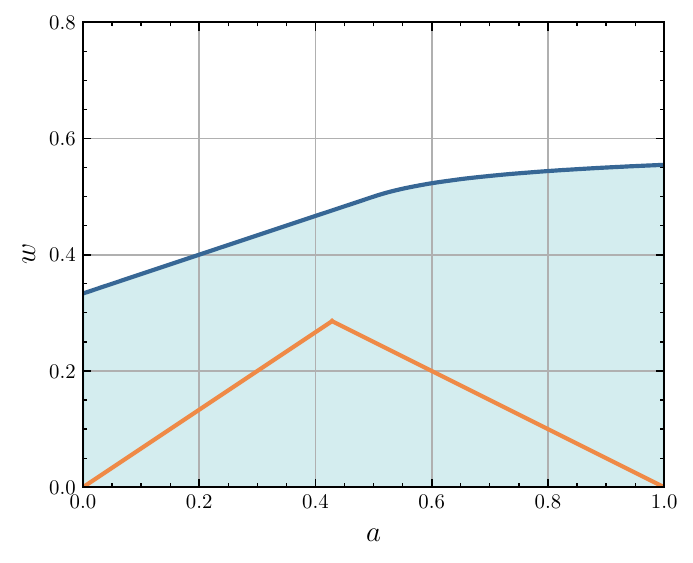}
 \\ \textbf{(b) $p=1.5$}
 \end{minipage}
 \caption{Positivity region for the map $\Phi$ in \cref{eq:map}, with $b=c=(1-a)/2$, $w,z\ge 0$, for two different ratios $p=z/w$. The blue roof represents the upper boundary of the positivity region, while the orange triangle marks complete positivity.}
 \label{fig:general_positivity}
\end{figure}

\subsection{\texorpdfstring{Case 3: $0\leq a=c\leq 1/2$, $b=1-2a$, and $w=z\geq 0$}{Case 3: $0\le a=c\le 1/2$, $b=1-2a$, $w=z\ge0$}}

\begin{proposition}
\label{prop:pos_case3}
Consider Case~3, namely $0\le a=c\le1/2$, $b=1-2a$, and $w=z\ge0$. Then the map $\Phi$ is positive if and only if
\begin{equation}
0\le w\le w_{\max}(a) \ ,
\qquad
w_{\max}(a)=
\begin{cases}
a+\sqrt{a-2a^2} \,, & 0\le a\le \tfrac13 \,,\quad a_0\le a\le \tfrac12 \,,\\[1mm]
w_{\mathrm{int}}(a) \,, & \tfrac13\le a\le a_0 \,,
\end{cases}
\label{eq:prop_wmax_case3}
\end{equation}
where $w_{\mathrm{int}}(a)$ is the interior branch described in \cref{app:caseCinterior}, and $a_0\in(1/3,1/2)$ is determined by
\[
w_{\mathrm{int}}(a_0)=a_0+\sqrt{a_0-2a_0^2}.
\]
Numerically, $a_0\simeq 0.43381$. Equivalently, the positivity boundary is the lower envelope of the boundary branch $a+\sqrt{a-2a^2}$ and the interior branch $w_{\mathrm{int}}(a)$.
\end{proposition}

\begin{proof}
Under these assumptions,
\begin{equation}
S(u_1,u_2)=a+(1-3a)u_2-w^2u_1\left(\frac{u_2}{a+(1-3a)(1-u_1-u_2)}+\frac{1-u_1-u_2}{a+(1-3a)u_1}\right) \,.
\label{eq:S_case3}
\end{equation}
The positivity boundary is obtained by finding the largest value $w_{\max}(a)$ such that $S(u_1,u_2)\ge0$ throughout the simplex.

Fix $u_1=t\in(0,1]$ and regard $S(t,u_2)$ as a function of $u_2\in[0,1-t]$. Its first and second derivatives are
\begin{align}
\frac{\partial S}{\partial u_2}
&= 1-3a - w^2 u_1 \left[ \frac{a+(1-3a)(1-u_1)}{\bigl(a+(1-3a)(1-u_1-u_2)\bigr)^2} - \frac{1}{a+(1-3a)u_1} \right] \,,
\label{eq:du_case3}\\
\frac{\partial^2 S}{\partial u_2^2}
&= -2w^2(1-3a)u_1 \frac{a+(1-3a)(1-u_1)}{\bigl(a+(1-3a)(1-u_1-u_2)\bigr)^3} \,.
\label{eq:d2u_case3}
\end{align}

\subparagraph*{\textbf{Subcase $\boldsymbol{0<a\le 1/3}$.}}
For fixed $t$, the second derivative is negative on the simplex, so the map $u_2\mapsto S(t,u_2)$ is concave. Its minimum over $u_2\in[0,1-t]$ is therefore attained at one of the endpoints, namely $u_2=0$ or $u_2=1-t$. Accordingly,
\begin{align}
S(t,0)
&= \frac{w^2 t^2 + \bigl[(1-3a)a-w^2\bigr]t + a^2}{(1-3a)t+a} \,,
\label{eq:S_case3_t0}\\[2mm]
S(t,1-t)
&= \frac{w^2 t^2 - \bigl[(1-3a)a+w^2\bigr]t + a-a^2}{a+(1-3a)t} \,.
\label{eq:S_case3_t1}
\end{align}
Requiring both boundary expressions to be non-negative for all $t\in[0,1]$ yields the condition
\begin{equation}
w_{\max}(a)=a+\sqrt{a-2a^2} \,,
\qquad 0 \le a\le\frac13 \,.
\label{eq:wmax_case3_left}
\end{equation}

\subparagraph*{\textbf{Subcase $\boldsymbol{1/3<a<1/2}$.}}
In this regime the minimum of $S(t,u_2)$ may occur either on the boundary, which again yields the branch~\eqref{eq:wmax_case3_left}, or at an interior stationary point. The latter possibility is analysed in detail in \cref{app:caseCinterior}. Let $w_{\mathrm{int}}(a)$ denote the smallest positive value of $w$ such that
$P\bigl(a,w^2\bigr)=0$,
where $P$ is given explicitly in \cref{app:caseCinterior}. Then the positivity boundary is
\begin{equation}
w_{\max}(a)=\min\left\{a+\sqrt{a-2a^2} \,,\,w_{\mathrm{int}}(a)\right\},
\qquad \frac13 \le a \le \frac12 \,.
\label{eq:wmax_case3_full}
\end{equation}

Introducing the crossing point $a_0\in(1/3,1/2)$ defined by
$w_{\mathrm{int}}(a_0)=a_0+\sqrt{a_0-2a_0^2}$,
we may therefore write
\begin{equation}
w_{\max}(a)=
\begin{cases}
a+\sqrt{a-2a^2} \,, & 0 \le a \le \tfrac13 \,, \quad a_0 \le a \le \tfrac12 \,,\\[1mm]
w_{\mathrm{int}}(a) \,, & \tfrac13 \le a \le a_0 \,,
\end{cases}
\label{eq:wmax_case3_piecewise}
\end{equation}
with numerical value $a_0\simeq 0.43381$. The map is positive if and only if $0\le w\le w_{\max}(a)$.

The complete-positivity region is defined by
\begin{equation}
w\le a \,, \qquad w\le \sqrt{a-2a^2} \,,
\label{53}
\end{equation}
corresponding to the area below the orange curve in \cref{fig:boundary_3}.
\begin{figure}[htb]
\includegraphics[width=0.4\textwidth]{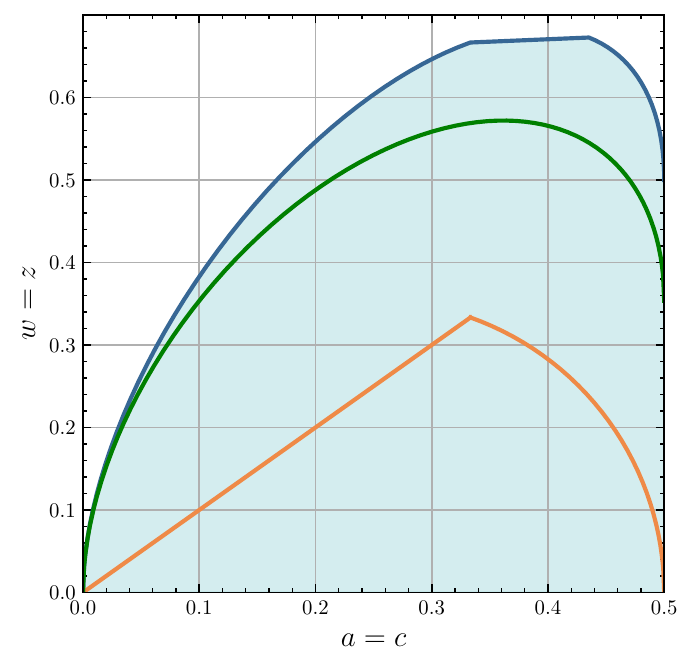}
\caption{Boundaries for different properties of the map in
  \cref{eq:map}, for the case $0 \leq a=c \leq1/2$, $b=1-2a$, and $w=z
  \geq 0$. The blue-shaded region below the blue curve indicates
  positivity [see \cref{eq:wmax_case3_piecewise}]; the orange curve
  delimits the complete-positivity region [see \cref{53}];
  the area below the green ellipse corresponds to decomposability [see
    \cref{eq:prop_dec_case3}].}
\label{fig:boundary_3}
\end{figure}
\end{proof}

\section{Decomposability thresholds and PPT witnesses}
\label{sec:decomposability}

We now characterize decomposability, i.e.\ determine when the map can
be written as the sum of a completely positive map and a completely
co-positive one. We recall that every $2$-positive map $
\Phi:\mathcal{B}(\mathbb{C}^3)\to \mathcal{B}(\mathbb{C}^3) $ is
decomposable~\cite{sanpera2001schmidt,yang2016all}. In the present
setting, however, we are concerned with maps that need only be
$1$-positive, and genuinely non-decomposable cases may therefore
occur.

A convenient criterion for detecting non-decomposability is obtained by testing the Choi matrix on PPT states. Let $C_\Phi$ be the Choi matrix associated with a positive map $\Phi$. If there exists a bipartite state $\rho$ such that
\[
\Tr[C_\Phi\,\rho]<0 \ ,
\]
then $\rho$ is necessarily entangled, because $C_\Phi$ is block-positive and therefore has non-negative expectation value on all separable states. If, in addition, $\rho$ is PPT, namely if $\rho^\tau\succeq0$ (where $\tau$ denotes partial transposition on one subsystem), then $\Phi$ must be non-decomposable.

Indeed, if $\Phi$ were decomposable, its Choi matrix could be written as
\begin{equation}
C_\Phi=P+Q^\tau\,,
\label{eq:decomp_choi}
\end{equation}
with $P,Q\succeq0$. For every PPT state $\rho$ one would then have
\begin{equation}
\Tr[C_\Phi\,\rho]
=
\Tr[P\rho]+\Tr[Q^\tau\rho]
=
\Tr[P\rho]+\Tr[Q\,\rho^\tau]
\ge 0\,,
\label{eq:ppt_test_decomp}
\end{equation}
because both $\rho$ and $\rho^\tau$ are positive semidefinite, and the trace of the product of two positive semidefinite operators is non-negative. Consequently, the existence of a PPT state $\rho$ such that $\Tr[C_\Phi\,\rho]<0$ rules out any decomposition of the form~\eqref{eq:decomp_choi}.

We therefore introduce the following PPT state of two qutrits, which serves as a convenient witness for the first two parametric families and, at a distinguished point, also for the third:
\begin{equation}
\rho _1= \frac{\sqrt{2}-1}{4}
\left(
\begin{array}{ccc|ccc|ccc}
2 & \cdot & \cdot & \cdot & \cdot & \cdot & \cdot & \cdot & -\sqrt{2} \\
\cdot & \sqrt{2} & \cdot & -\sqrt{2} & \cdot & \cdot & \cdot & \cdot & \cdot \\
\cdot & \cdot & \sqrt{2} & \cdot & \cdot & \cdot & \cdot & \cdot & \cdot \\ \hline
\cdot & -\sqrt{2} & \cdot & \sqrt{2} & \cdot & \cdot & \cdot & \cdot & \cdot \\
\cdot & \cdot & \cdot & \cdot & 1 & \cdot & \cdot & \cdot & \cdot \\
\cdot & \cdot & \cdot & \cdot & \cdot & \cdot & \cdot & \cdot & \cdot \\ \hline
\cdot & \cdot & \cdot & \cdot & \cdot & \cdot & \sqrt{2} & \cdot & \cdot \\
\cdot & \cdot & \cdot & \cdot & \cdot & \cdot & \cdot & \cdot & \cdot \\
-\sqrt{2} & \cdot & \cdot & \cdot & \cdot & \cdot & \cdot & \cdot & 1
\end{array}
\right)\,.
\label{eq:ppt}
\end{equation}
Its spectrum is
\begin{equation}
\sigma(\rho _1) =
\left\{
\frac{3}{4}(\sqrt{2}-1),\,
\frac{1}{2}(2-\sqrt{2}),\,
\frac{1}{4}(2-\sqrt{2}),\,
\frac{1}{4}(2-\sqrt{2}),\,
\frac{1}{4}(\sqrt{2}-1),\,
0,\,
0,\,
0,\,
0
\right\}\,.
\label{eq:spectrum_rho}
\end{equation}
Henceforth, the partial transposition $\tau$ is understood to act on
the second tensor factor. The partially transposed state $\rho_1^\tau$
has the same spectrum of $\rho_1$, and therefore the state is PPT.

We now reconsider the three parametric families introduced in
\cref{sec:positivity} and determine, in each case, the boundary
between decomposable and non-decomposable maps.

\subsection{Case 1: \texorpdfstring{$a\geq 0$, $b=c\geq 0$ and $w=z=1$}{Case 1: $a\ge0$, $b=c\ge0$, $w=z=1$}}

\begin{proposition}
\label{prop:dec_case1}
In Case~1, the decomposability threshold inside the positivity region is
\begin{equation}
b_{\mathrm{dec}}(a)=1-\frac{a}{\sqrt{2}}\,.
\label{eq:prop_dec_case1}
\end{equation}
Equivalently, the map is decomposable if and only if $b\ge b_{\mathrm{dec}}(a)$, and it is non-decomposable whenever $b<b_{\mathrm{dec}}(a)$.
\end{proposition}

\begin{proof}
Let us denote by $C^{(1)}_{a,b}$ the Choi matrix corresponding to the choice $b=c$ and $w=z=1$. Evaluating it on the PPT state \cref{eq:ppt} gives
\begin{equation}
\Tr \bigl[ C^{(1)}_{a,b}\,\rho _1\bigr]=(\sqrt{2}-1)(a+\sqrt{2}b-\sqrt{2})\,.
\label{eq:choi_action1}
\end{equation}
Therefore, $\rho_1$ is PPT entangled, and the map is non-decomposable whenever
\[
\Tr \bigl[C^{(1)}_{a,b}\rho_1\bigr]<0 \ ,
\]
namely for
\[
b<1-\frac{a}{\sqrt{2}}\,.
\]

We now show that this is the exact boundary separating decomposable
and non-decomposable maps. The argument is based on convexity. Any
pair of parameters $(a',b')$ satisfying $b'\ge 1-a'/\sqrt{2}$ yields a
Choi matrix that can be written as a convex combination of two
extremal matrices:
\begin{equation}
C^{(1)}_{a',b'} = \lambda\, C^{(1)}_{a,0} + (1-\lambda)\, C^{(1)}_{0,b}\,,
\quad 0 \le \lambda \le 1\,,
\end{equation}
with $a\ge \sqrt{2}$ and $b\ge 1$. It is therefore enough to establish decomposability for the two extremal families $C^{(1)}_{a,0}$ and $C^{(1)}_{0,b}$.

For the first family one readily verifies that $C^{(1)}_{a,0}=C_1 +
C_2 ^\tau $, where $C_1$, $C_2$ are positive matrices for $a\ge
\sqrt{2}$, explicitly given by
\begin{equation}
C_1 =
\left(
\begin{array}{ccc|ccc|ccc}
a/2 & \cdot & \cdot & \cdot & \cdot & \cdot & \cdot & \cdot & 1 \\
\cdot & \cdot & \cdot & \cdot & \cdot & \cdot & \cdot & \cdot & \cdot \\
\cdot & \cdot & \cdot & \cdot & \cdot & \cdot & \cdot & \cdot & \cdot \\ \hline
\cdot & \cdot & \cdot & \cdot & \cdot & \cdot & \cdot & \cdot & \cdot \\
\cdot & \cdot & \cdot & \cdot & \cdot & \cdot & \cdot & \cdot & \cdot \\
\cdot & \cdot & \cdot & \cdot & \cdot & \cdot & \cdot & \cdot & \cdot \\ \hline
\cdot & \cdot & \cdot & \cdot & \cdot & \cdot & \cdot & \cdot & \cdot \\
\cdot & \cdot & \cdot & \cdot & \cdot & \cdot & \cdot & \cdot & \cdot \\
1 & \cdot & \cdot & \cdot & \cdot & \cdot & \cdot & \cdot & a
\end{array}
\right)\,,
\quad
C_2 =
\left(
\begin{array}{ccc|ccc|ccc}
a/2 & \cdot & \cdot & \cdot & 1 & \cdot & \cdot & \cdot & \cdot \\
\cdot & \cdot & \cdot & \cdot & \cdot & \cdot & \cdot & \cdot & \cdot \\
\cdot & \cdot & \cdot & \cdot & \cdot & \cdot & \cdot & \cdot & \cdot \\ \hline
\cdot & \cdot & \cdot & \cdot & \cdot & \cdot & \cdot & \cdot & \cdot \\
1 & \cdot & \cdot & \cdot & a & \cdot & \cdot & \cdot & \cdot \\
\cdot & \cdot & \cdot & \cdot & \cdot & \cdot & \cdot & \cdot & \cdot \\ \hline
\cdot & \cdot & \cdot & \cdot & \cdot & \cdot & \cdot & \cdot & \cdot \\
\cdot & \cdot & \cdot & \cdot & \cdot & \cdot & \cdot & \cdot & \cdot \\
\cdot & \cdot & \cdot & \cdot & \cdot & \cdot & \cdot & \cdot & \cdot
\end{array}
\right)
\label{eq:c1c2_1}\,.
\end{equation}

Similarly $C^{(1)}_{0,b}=\tilde C_1 + \tilde C_2 ^\tau $, where
$\tilde C_1$ and $\tilde C_2$ are positive matrices for $b\ge 1$,
given by
\begin{equation}
\tilde C_1 =
\left(
\begin{array}{ccc|ccc|ccc}
\cdot & \cdot & \cdot & \cdot & \cdot & \cdot & \cdot & \cdot & \cdot \\
\cdot & b & \cdot & 1 & \cdot & \cdot & \cdot & \cdot & \cdot \\
\cdot & \cdot & \cdot & \cdot & \cdot & \cdot & \cdot & \cdot & \cdot \\ \hline
\cdot & 1 & \cdot & b & \cdot & \cdot & \cdot & \cdot & \cdot \\
\cdot & \cdot & \cdot & \cdot & \cdot & \cdot & \cdot & \cdot & \cdot \\
\cdot & \cdot & \cdot & \cdot & \cdot & b & \cdot & \cdot & \cdot \\ \hline
\cdot & \cdot & \cdot & \cdot & \cdot & \cdot & \cdot & \cdot & \cdot \\
\cdot & \cdot & \cdot & \cdot & \cdot & \cdot & \cdot & b & \cdot \\
\cdot & \cdot & \cdot & \cdot & \cdot & \cdot & \cdot & \cdot & \cdot
\end{array}
\right)\,,
\quad
\tilde C_2 =
\left(
\begin{array}{ccc|ccc|ccc}
\cdot & \cdot & \cdot & \cdot & \cdot & \cdot & \cdot & \cdot & \cdot \\
\cdot & \cdot & \cdot & \cdot & \cdot & \cdot & \cdot & \cdot & \cdot \\
\cdot & \cdot & b & \cdot & \cdot & \cdot & 1 & \cdot & \cdot \\ \hline
\cdot & \cdot & \cdot & \cdot & \cdot & \cdot & \cdot & \cdot & \cdot \\
\cdot & \cdot & \cdot & \cdot & \cdot & \cdot & \cdot & \cdot & \cdot \\
\cdot & \cdot & \cdot & \cdot & \cdot & \cdot & \cdot & \cdot & \cdot \\ \hline
\cdot & \cdot & 1 & \cdot & \cdot & \cdot & b & \cdot & \cdot \\
\cdot & \cdot & \cdot & \cdot & \cdot & \cdot & \cdot & \cdot & \cdot \\
\cdot & \cdot & \cdot & \cdot & \cdot & \cdot & \cdot & \cdot & \cdot
\end{array}
\right)\,.
\label{eq:c1c2tilde_1}
\end{equation}

It follows that the map is decomposable throughout the region $b\ge 1-a/\sqrt{2}$, while it is non-decomposable below that line. Therefore, within the positivity region, the map is decomposable if and only if $b\ge 1-a/\sqrt{2}$.
\end{proof}

\subsection{Case 2: \texorpdfstring{$0\leq a\leq 1$, $b=c=(1-a)/2$ and $w=z\geq 0$}{Case 2: $0\le a\le 1$, $b=c=(1-a)/2$, $w=z\ge0$}}

\begin{proposition}
\label{prop:dec_case2}
In Case~2, the decomposability threshold inside the positivity region is
\begin{equation}
w_{\mathrm{dec}}(a)=\frac{\sqrt{2}-1}{2}a+\frac12\,.
\label{eq:prop_dec_case2}
\end{equation}
Equivalently, the map is decomposable if and only if $w\le w_{\mathrm{dec}}(a)$, and it is non-decomposable whenever $w> w_{\mathrm{dec}}(a)$.
\end{proposition}

\begin{proof}
Let us denote by $C^{(2)}_{a,w}$ the corresponding Choi
matrix. Evaluating the map with $b=c=(1-a)/2$ and $w=z\ge 0$ on the
bound-entangled state \eqref{eq:ppt} gives
\begin{equation}
\Tr\left[ C^{(2)}_{a,w}\,\rho _1 \right]
=(\sqrt{2}-2)\left( w-\frac{1}{2}-\frac{\sqrt{2}-1}{2}\,a \right)\,.
\end{equation}
Therefore, the map is non-decomposable whenever
\[
w>\frac{\sqrt{2}-1}{2}a+\frac12 \ .
\]

This criterion is in fact sharp, and the proof again relies on convexity. Assume that
\[
w'\leq \frac{\sqrt{2}-1}{2}a' + \frac12 \ .
\]
Then any $C^{(2)}_{a',w'}$ can be written as a convex combination of matrices of the form $C^{(2)}_{0,w}$ with $w\leq 1/2$ and $C^{(2)}_{1,w}$ with $w \leq \sqrt{2}/2$. It is therefore enough to establish decomposability for these two extremal families.

Indeed, $ C^{(2)}_{0,w}=C_1 + C_2 ^\tau$, where $C_1$ and $C_2$ are
positive matrices for $w\leq 1/2$, given by
\begin{equation}
C_1 =
\left(
\begin{array}{ccc|ccc|ccc}
\cdot & \cdot & \cdot & \cdot & \cdot & \cdot & \cdot & \cdot & \cdot \\
\cdot & 1/2 & \cdot & w & \cdot & \cdot & \cdot & \cdot & \cdot \\
\cdot & \cdot & \cdot & \cdot & \cdot & \cdot & \cdot & \cdot & \cdot \\ \hline
\cdot & w & \cdot & 1/2 & \cdot & \cdot & \cdot & \cdot & \cdot \\
\cdot & \cdot & \cdot & \cdot & \cdot & \cdot & \cdot & \cdot & \cdot \\
\cdot & \cdot & \cdot & \cdot & \cdot & 1/2 & \cdot & \cdot & \cdot \\ \hline
\cdot & \cdot & \cdot & \cdot & \cdot & \cdot & \cdot & \cdot & \cdot \\
\cdot & \cdot & \cdot & \cdot & \cdot & \cdot & \cdot & 1/2 & \cdot \\
\cdot & \cdot & \cdot & \cdot & \cdot & \cdot & \cdot & \cdot & \cdot
\end{array}
\right)\,,
\quad
C_2 =
\left(
\begin{array}{ccc|ccc|ccc}
\cdot & \cdot & \cdot & \cdot & \cdot & \cdot & \cdot & \cdot & \cdot \\
\cdot & \cdot & \cdot & \cdot & \cdot & \cdot & \cdot & \cdot & \cdot \\
\cdot & \cdot & 1/2 & \cdot & \cdot & \cdot & w & \cdot & \cdot \\ \hline
\cdot & \cdot & \cdot & \cdot & \cdot & \cdot & \cdot & \cdot & \cdot \\
\cdot & \cdot & \cdot & \cdot & \cdot & \cdot & \cdot & \cdot & \cdot \\
\cdot & \cdot & \cdot & \cdot & \cdot & \cdot & \cdot & \cdot & \cdot \\ \hline
\cdot & \cdot & w & \cdot & \cdot & \cdot & 1/2 & \cdot & \cdot \\
\cdot & \cdot & \cdot & \cdot & \cdot & \cdot & \cdot & \cdot & \cdot \\
\cdot & \cdot & \cdot & \cdot & \cdot & \cdot & \cdot & \cdot & \cdot
\end{array}
\right)\,.
\end{equation}

Similarly, $C^{(2)}_{1,w}=\tilde C_1 + \tilde C_2 ^\tau$, where
$\tilde C_1$ and $\tilde C_2$ are positive matrices for $w\leq
\sqrt{2}/2$, given by
\begin{equation}
\tilde C_1 =
\left(
\begin{array}{ccc|ccc|ccc}
1/2 & \cdot & \cdot & \cdot & \cdot & \cdot & \cdot & \cdot & w \\
\cdot & \cdot & \cdot & \cdot & \cdot & \cdot & \cdot & \cdot & \cdot \\
\cdot & \cdot & \cdot & \cdot & \cdot & \cdot & \cdot & \cdot & \cdot \\ \hline
\cdot & \cdot & \cdot & \cdot & \cdot & \cdot & \cdot & \cdot & \cdot \\
\cdot & \cdot & \cdot & \cdot & \cdot & \cdot & \cdot & \cdot & \cdot \\
\cdot & \cdot & \cdot & \cdot & \cdot & \cdot & \cdot & \cdot & \cdot \\ \hline
\cdot & \cdot & \cdot & \cdot & \cdot & \cdot & \cdot & \cdot & \cdot \\
\cdot & \cdot & \cdot & \cdot & \cdot & \cdot & \cdot & \cdot & \cdot \\
w & \cdot & \cdot & \cdot & \cdot & \cdot & \cdot & \cdot & 1/2
\end{array}
\right)\,,
\quad
\tilde C_2 =
\left(
\begin{array}{ccc|ccc|ccc}
1/2 & \cdot & \cdot & \cdot & w & \cdot & \cdot & \cdot & \cdot \\
\cdot & \cdot & \cdot & \cdot & \cdot & \cdot & \cdot & \cdot & \cdot \\
\cdot & \cdot & \cdot & \cdot & \cdot & \cdot & \cdot & \cdot & \cdot \\ \hline
\cdot & \cdot & \cdot & \cdot & \cdot & \cdot & \cdot & \cdot & \cdot \\
w & \cdot & \cdot & \cdot & 1 & \cdot & \cdot & \cdot & \cdot \\
\cdot & \cdot & \cdot & \cdot & \cdot & \cdot & \cdot & \cdot & \cdot \\ \hline
\cdot & \cdot & \cdot & \cdot & \cdot & \cdot & \cdot & \cdot & \cdot \\
\cdot & \cdot & \cdot & \cdot & \cdot & \cdot & \cdot & \cdot & \cdot \\
\cdot & \cdot & \cdot & \cdot & \cdot & \cdot & \cdot & \cdot & 1/2
\end{array}
\right)\,.
\end{equation}

It follows that the map is decomposable in the whole region
\[
w\leq \frac{\sqrt{2}-1}{2}a + \frac12 \ ,
\]
and non-decomposable above it. Therefore, within the positivity region, the map is decomposable if and only if $w\leq (\sqrt{2}-1)a/2 + 1/2$.
\end{proof}

\subparagraph{Eigenvalue bound.}
We now interpret the trace-preserving family of Case~2 in the light of
Theorem~1 of Ref.~\cite{boundoneigs2025}, which states that every
$2$-positive trace-preserving map $\Phi:\mathcal{B}(\mathbb{C}^d)\to\mathcal{B}(\mathbb{C}^d)$ satisfies
\begin{equation}
\Tr \Phi \le d \,\min \Re[\sigma(\Phi)] + d^2 - d\,.
\label{eq:boundary}
\end{equation}
Since in the present case $a+b+c=1$, the maps are trace-preserving,
and the bound applies whenever they are $2$-positive. In our recent
work~\cite{morgillo2026maps} we observed numerically that, unlike the
previously known $d=3$ examples discussed in~\cite{boundoneigs2025},
this family contains positive maps that
violate~\eqref{eq:boundary}. Such violations can occur both in the
decomposable and in the non-decomposable region.

For the present family one has
\begin{equation}
\sigma(\Phi)=\left \{ 1, \frac{3a-1}{2}, \frac{3a-1}{2}, w, w, w, -w, 0, 0 \right \}\,,
\label{eq:spectrum}
\end{equation}
so that
\begin{equation}
\Tr\Phi = 3a+2w\,,
\qquad
\min\Re[\sigma(\Phi)] = \min\left\{\frac{3a-1}{2},-w\right\}\,.
\end{equation}
Substituting these expressions into~\eqref{eq:boundary} gives
\begin{equation}
3a + 2w - 6 - 3 \min\left( \frac{3a - 1}{2},\, -w \right) \leq 0
\,.
\label{eq:bb}
\end{equation}
Therefore, \eqref{eq:bb} is a necessary condition for $2$-positivity in the parameter space.

The red triangle in \cref{fig:boundary2} marks a region containing maps that violate the spectral bound~\eqref{eq:boundary}. Hence any map in the red triangle is necessarily positive but not $2$-positive. In this sense, the violation provides a purely spectral witness of the gap between $1$-positivity and $2$-positivity. In the present family this statement is even sharper, because, as noted in \cref{rem1}, $2$-positivity already implies complete positivity. Thus the red triangle isolates a sector of the positive region that is automatically excluded from the completely positive one, while still intersecting both the decomposable and the genuinely non-decomposable parts of the $1$-positive parameter space.

\subsection{Case 3: $0\leq a=c\leq 1/2$, $b=1-2a$, and $w=z\geq 0$}

\begin{proposition}
\label{prop:dec_case3}
In Case~3, the decomposability threshold inside the positivity region is
\begin{equation}
w_{\mathrm{dec}}(a)=\frac{a}{\sqrt{2}}+\sqrt{a-2a^2}\,.
\label{eq:prop_dec_case3}
\end{equation}
Equivalently, the map is decomposable if and only if $w\le w_{\mathrm{dec}}(a)$, and it is non-decomposable whenever $w> w_{\mathrm{dec}}(a)$.
\end{proposition}

\begin{proof}
  Let us denote by $C^{(3)}_{a,w}$ the corresponding Choi matrix. We
  first introduce two PPT states, both of
  which are special cases of the broader family analysed later in
  \cref{sec:bound_ent_states}. The first state is
\begin{equation}
\rho_2=\frac{1}{10}
\left(
\begin{array}{ccc|ccc|ccc}
2 & \cdot & \cdot & \cdot & \cdot & \cdot & \cdot & \cdot & -\sqrt{2} \\
\cdot & 1 & \cdot & -\sqrt{2} & \cdot & \cdot & \cdot & \cdot & \cdot \\
\cdot & \cdot & 2 & \cdot & \cdot & \cdot & \cdot & \cdot & \cdot \\ \hline
\cdot & -\sqrt{2} & \cdot & 2 & \cdot & \cdot & \cdot & \cdot & \cdot \\
\cdot & \cdot & \cdot & \cdot & 1 & \cdot & \cdot & \cdot & \cdot \\
\cdot & \cdot & \cdot & \cdot & \cdot & \cdot & \cdot & \cdot & \cdot \\ \hline
\cdot & \cdot & \cdot & \cdot & \cdot & \cdot & 1 & \cdot & \cdot \\
\cdot & \cdot & \cdot & \cdot & \cdot & \cdot & \cdot & \cdot & \cdot \\
-\sqrt{2} & \cdot & \cdot & \cdot & \cdot & \cdot & \cdot & \cdot & 1
\end{array}
\right)\,.
\label{eq:rho1}
\end{equation}
The eigenvalues are given by 
\[ \sigma (\rho_2)= \sigma (\rho_2^\tau )=
\{3/10,3/10,1/5,1/10,1/10,0,0,0,0\} \ .
\]
Moreover,
\begin{equation}
\Tr[C^{(3)}_{a,w} \rho_2] = \frac 15 (2a+1-2\sqrt{2} w)\,.
\end{equation}
The solution of $\Tr[C^{(3)}_{a,w}\rho_2] = 0$ is
\[
w=\frac{2a+1}{2\sqrt{2}}\,.
\]
This line is tangent to the ellipse in \cref{eq:prop_dec_case3} at the point
\[
a = \frac14 \ ,
\qquad
w = \frac{3\sqrt{2}}{8} \ .
\]

The second state is
\begin{equation}
\rho_3=\frac{1}{10}
\left(
\begin{array}{ccc|ccc|ccc}
2 & \cdot & \cdot & \cdot & \cdot & \cdot & \cdot & \cdot & -\sqrt{2} \\
\cdot & 2 & \cdot & -\sqrt{2} & \cdot & \cdot & \cdot & \cdot & \cdot \\
\cdot & \cdot & 1 & \cdot & \cdot & \cdot & \cdot & \cdot & \cdot \\ \hline
\cdot & -\sqrt{2} & \cdot & 1 & \cdot & \cdot & \cdot & \cdot & \cdot \\
\cdot & \cdot & \cdot & \cdot & 1 & \cdot & \cdot & \cdot & \cdot \\
\cdot & \cdot & \cdot & \cdot & \cdot & \cdot & \cdot & \cdot & \cdot \\ \hline
\cdot & \cdot & \cdot & \cdot & \cdot & \cdot & 2 & \cdot & \cdot \\
\cdot & \cdot & \cdot & \cdot & \cdot & \cdot & \cdot & \cdot & \cdot \\
-\sqrt{2} & \cdot & \cdot & \cdot & \cdot & \cdot & \cdot & \cdot & 1
\end{array}
\right)\,.
\label{eq:rho2}
\end{equation}
For this state,
\begin{equation}
\Tr[C^{(3)}_{a,w}\rho_3] = \frac 15 (2-a-2\sqrt{2} w )\,.
\end{equation}
The solution of $\Tr[C^{(3)}_{a,w}\rho_3]=0$ is
\[
w=\frac{2-a}{2\sqrt{2}}\,.
\]
This line is tangent to the ellipse in \cref{eq:prop_dec_case3} at the point
\[
a = \frac25 \ ,
\qquad
w = \frac{2\sqrt{2}}{5} \ .
\]

More generally, for each point on the ellipse, one can construct a rank-four bound-entangled state of the form
\begin{equation}
\rho (a)=
\left(
\begin{array}{ccc|ccc|ccc}
q & \cdot & \cdot & \cdot & \cdot & \cdot & \cdot & \cdot & -g \\
\cdot & r & \cdot & -g & \cdot & \cdot & \cdot & \cdot & \cdot \\
\cdot & \cdot & s & \cdot & \cdot & \cdot & \cdot & \cdot & \cdot \\ \hline
\cdot & -g & \cdot & s & \cdot & \cdot & \cdot & \cdot & \cdot \\
\cdot & \cdot & \cdot & \cdot & v & \cdot & \cdot & \cdot & \cdot \\
\cdot & \cdot & \cdot & \cdot & \cdot & \cdot & \cdot & \cdot & \cdot \\ \hline
\cdot & \cdot & \cdot & \cdot & \cdot & \cdot & r & \cdot & \cdot \\
\cdot & \cdot & \cdot & \cdot & \cdot & \cdot & \cdot & \cdot & \cdot \\
-g & \cdot & \cdot & \cdot & \cdot & \cdot & \cdot & \cdot & v
\end{array}
\right)\,,
\label{eq:bound_entan}
\end{equation}
such that
\begin{equation}
\Tr\Big[ C^{(3)}_{a, \frac{a}{\sqrt{2}} + \sqrt{a - 2a^2}} \, \rho (a)\Big] = 0\,,
\qquad
\partial_a \Tr\Big[ C^{(3)}_{a, \frac{a}{\sqrt{2}} + \sqrt{a - 2a^2}} \, \rho (a)\Big] = 0\,,
\end{equation}
and the parameters $q,r,s,v,g$ are positive, satisfying
\[
g = \sqrt{rs} = \sqrt{qv} \ .
\]

Analytical expressions for the matrix elements as functions of $a$ are obtained by solving
\begin{equation}
\begin{cases}
4\!\left(\frac{a}{\sqrt{2}}+\sqrt{a-2a^{2}}\right) g -a (q+2s+2v) - 2(1-2a)r =0\,, \\
4\!\left(\frac{1}{\sqrt{2}}+\frac{1-4a}{2\sqrt{a-2a^{2}}}\right) g -q-2s-2v +4r=0\,, \\
 g =\sqrt{q\,v}\,, \\
 q\,v = r\,s\,, \\
 q+ 2 (r + s + v)=1\,.
\end{cases}
\end{equation}
The solutions are
\begin{equation}
\begin{aligned}
r &= \frac{a(1-a)-\sqrt{2}\,a \sqrt{a-2a^2}}{ 2\,(1 - 4a + 5a^{2})}\,, &
s &= \frac{(1 - \sqrt{2}\,\sqrt{a-2a^2} - a)(1 - 2a)}{ 2\,(1 - 4a + 5a^{2})}\,, &
v &= \frac{\sqrt{2}\,\sqrt{a-2a^2}(1-a) - 2a + 4a^{2}}{ 4\,(1 - 4a + 5a^{2})}\,, \\
q &= 2 v\,, &
g &= \sqrt{2}\,v\,.
\label{eq:sol}
\end{aligned}
\end{equation}
These matrix elements are plotted as functions of the parameter $a$ in \cref{fig:mat_elems}. In particular, one recovers the bound-entangled states in \cref{eq:ppt,eq:rho1,eq:rho2} for $a = 1/3$, $1/4$, and $2/5$, respectively.
\begin{figure}[hbt]
	\centering
	\includegraphics[width=0.55\textwidth]{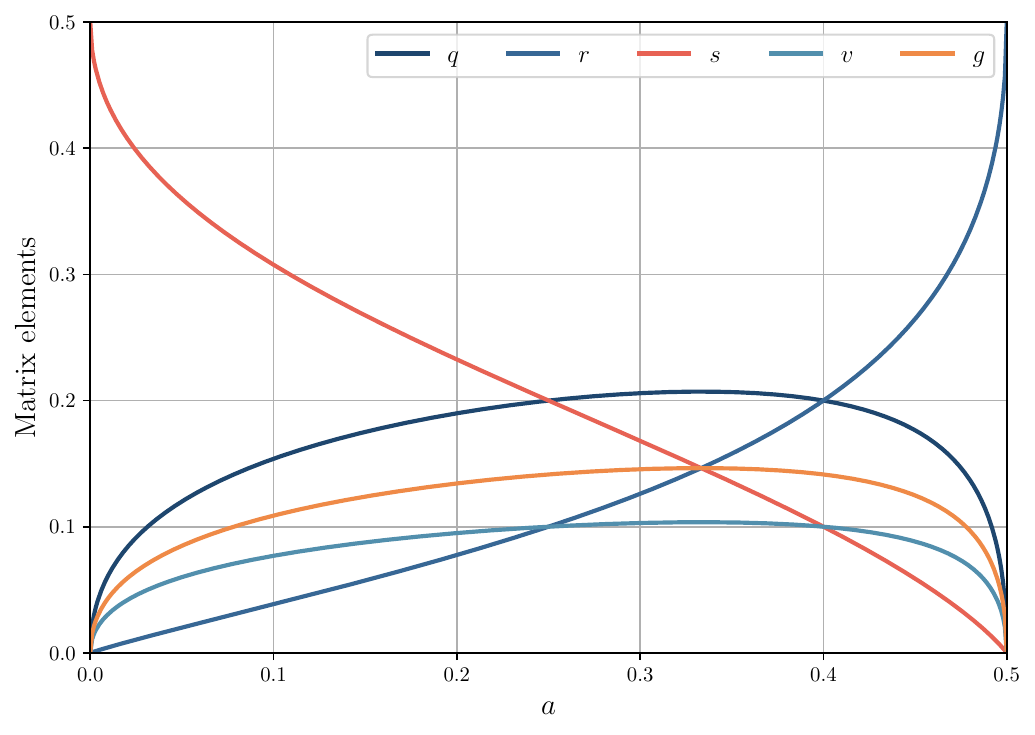}
	\caption[Bound entangled states characterization]{Values of the matrix elements of the rank-four bound-entangled state in \cref{eq:bound_entan} as functions of the parameter $a$.}
	\label{fig:mat_elems}
\end{figure}

Thus we proved that the map is non-decomposable for
\[
w > \frac{a}{\sqrt{2}} + \sqrt{a-2a^2} \ .
\]
We next prove the converse statement, namely that
\cref{eq:prop_dec_case3} is also the decomposability boundary. To this
end, one can write
\begin{equation}
 C^{(3)}_{a,\frac a {\sqrt 2}+\sqrt {a-2a^2}}
 = C_1 + C_2^\tau\,,
\end{equation}
where
\begin{equation}
\scalebox{0.9}{$
C_1 =
\left(
\begin{array}{ccc|ccc|ccc}
\frac{a}{2} & \cdot & \cdot & \cdot & \cdot & \cdot & \cdot & \cdot & \frac{a}{\sqrt{2}} \\
\cdot & 1-2a & \cdot & \sqrt{a-2a^2} & \cdot & \cdot & \cdot & \cdot & \cdot \\
\cdot & \cdot & \cdot & \cdot & \cdot & \cdot & \cdot & \cdot & \cdot \\ \hline
\cdot & \sqrt{a-2a^2} & \cdot & a & \cdot & \cdot & \cdot & \cdot & \cdot \\
\cdot & \cdot & \cdot & \cdot & \cdot & \cdot & \cdot & \cdot & \cdot \\
\cdot & \cdot & \cdot & \cdot & \cdot & \frac{1-2a}{2} & \cdot & \cdot & \cdot \\ \hline
\cdot & \cdot & \cdot & \cdot & \cdot & \cdot & \cdot & \cdot & \cdot \\
\cdot & \cdot & \cdot & \cdot & \cdot & \cdot & \cdot & \frac{a}{2} & \cdot \\
\frac{a}{\sqrt{2}} & \cdot & \cdot & \cdot & \cdot & \cdot & \cdot & \cdot & a
\end{array}
\right)\,,\quad
C_2 =
\left(
\begin{array}{ccc|ccc|ccc}
\frac{a}{2} & \cdot & \cdot & \cdot & \frac{a}{\sqrt{2}} & \cdot & \cdot & \cdot & \cdot \\
\cdot & \cdot & \cdot & \cdot & \cdot & \cdot & \cdot & \cdot & \cdot \\
\cdot & \cdot & a & \cdot & \cdot & \cdot & \sqrt{a-2a^2} & \cdot & \cdot \\ \hline
\cdot & \cdot & \cdot & \cdot & \cdot & \cdot & \cdot & \cdot & \cdot \\
\frac{a}{\sqrt{2}} & \cdot & \cdot & \cdot & a & \cdot & \cdot & \cdot & \cdot \\
\cdot & \cdot & \cdot & \cdot & \cdot & \frac{1-2a}{2} & \cdot & \cdot & \cdot \\ \hline
\cdot & \cdot & \sqrt{a-2a^2} & \cdot & \cdot & \cdot & 1-2a & \cdot & \cdot \\
\cdot & \cdot & \cdot & \cdot & \cdot & \cdot & \cdot & \frac{a}{2} & \cdot \\
\cdot & \cdot & \cdot & \cdot & \cdot & \cdot & \cdot & \cdot & \cdot
\end{array}
\right) $}\,,
\end{equation}
both of which are positive.

Therefore, the map is decomposable within the positivity region if and only if
\begin{equation}
  w\le \frac{a}{\sqrt{2}} + \sqrt{a-2a^2} \ .
\label{eq:ww1}
\end{equation}
Indeed, any map satisfying \cref{eq:ww1} is decomposable, 
because it can be expressed as a convex combination of the decomposable boundary map and the completely positive map $C^{(3)}_{a,0}$.
\end{proof}
\begin{remark}
Note that \cref{eq:bb} still holds, but in the present case no violation is found.
\end{remark}

\section{Associated PPT and bound-entangled state families}
\label{sec:bound_ent_states}

We now turn from the map side to the associated bipartite states. The
aim of this section is twofold. First, we introduce witness-adapted
deformations of the PPT states used in \cref{sec:decomposability};
these families make the transition from separability to PPT
entanglement explicit, but their PPT-entangled branch is not always
detected optimally by the simplest fixed witnesses. Second, we
introduce alternative one-parameter families, whose full PPT-entangled
branch is detected by fixed positive maps. The first class is
therefore naturally tied to the decomposability analysis, whereas the
second sharpens the separability threshold.

\subsection{Original witness-adapted families}

All one-parameter deformations introduced below remain PPT up to
$k=\sqrt2$; the nontrivial issue is where separability ends and PPT
entanglement begins, and how efficiently that transition is detected by the
witnesses arising from \cref{sec:decomposability}.

\subparagraph{First family.}
\begin{proposition}
\label{prop:rho1k}
The family $\rho_1(k)$ defined in \cref{eq:rho1k} has the following
properties:
\begin{enumerate}
\item $\rho_1(k)$ is PPT for all $0\le k\le \sqrt2$;
\item the fixed witness $C^{(1)}_{1/2,\,1/2}$ detects entanglement whenever
\[
\frac{1+\sqrt2}{2}<k\le \sqrt2\,;
\]
\item the state $\rho_1(1)$ is separable and admits the explicit product
decomposition given in \cref{eq:rho1-sep};
\item according to the level-$2$ DPS hierarchy with PPT constraints,
$\rho_1(k)$ is entangled for every $k>1$.
\end{enumerate}
\end{proposition}

\begin{proof}
Consider the family
\begin{equation}
\rho_1(k)= \frac{\sqrt{2}-1}{4}
\left(
\begin{array}{ccc|ccc|ccc}
2 & \cdot & \cdot & \cdot & \cdot & \cdot & \cdot & \cdot & -k \\
\cdot & \sqrt{2} & \cdot & -k & \cdot & \cdot & \cdot & \cdot & \cdot \\
\cdot & \cdot & \sqrt{2} & \cdot & \cdot & \cdot & \cdot & \cdot & \cdot \\ \hline
\cdot & -k & \cdot & \sqrt{2} & \cdot & \cdot & \cdot & \cdot & \cdot \\
\cdot & \cdot & \cdot & \cdot & 1 & \cdot & \cdot & \cdot & \cdot \\
\cdot & \cdot & \cdot & \cdot & \cdot & \cdot & \cdot & \cdot & \cdot \\ \hline
\cdot & \cdot & \cdot & \cdot & \cdot & \cdot & \sqrt{2} & \cdot & \cdot \\
\cdot & \cdot & \cdot & \cdot & \cdot & \cdot & \cdot & \cdot & \cdot \\
-k & \cdot & \cdot & \cdot & \cdot & \cdot & \cdot & \cdot & 1
\end{array}
\right),
\label{eq:rho1k}
\end{equation}
which is PPT for $0\le k\le \sqrt2$.

The Case~1 witness family gives
\begin{equation}
\Tr\!\bigl[C^{(1)}_{a,b}\,\rho_1(k)\bigr]
=
(\sqrt{2}-1)\bigl(a+\sqrt2\,b-k\bigr) \ .
\label{eq:rho1k-trace}
\end{equation}
Choosing $a=b=1/2$ yields
\[
\Tr\!\bigl[C^{(1)}_{1/2,\,1/2}\,\rho_1(k)\bigr]=0
\qquad\Longleftrightarrow\qquad
k=\frac{1+\sqrt2}{2} \ ,
\]
so $C^{(1)}_{1/2,\,1/2}$ detects entanglement for
$\frac{1+\sqrt2}{2}<k\le \sqrt2$.

For smaller values of $k$, stronger criteria are needed. The CCNR
criterion~\cite{chen2002matrix,rudolph2005further} proves entanglement when
\begin{equation}
k>\sqrt{2}+\frac34-\frac{\sqrt{17}}{4}\simeq 1.1334 \ ,
\end{equation}
while a level-$2$ DPS test with PPT constraints~\cite{doherty2004complete}
indicates entanglement for every $k>1$.

At $k=1$, the state is separable and admits the decomposition
\begin{equation}
\rho_1(1)=\sum_{i=1}^{12} p_i\,\ket{\psi_i}\!\bra{\psi_i} \ ,
\label{eq:rho1-sep}
\end{equation}
with product vectors
\begin{align}
&\ket{\psi_1}= \tfrac12 (\ket{1}+\ket{3})\otimes(\ket{1}-\ket{3}) \ ,
&&\ket{\psi_2}= \tfrac12 (\ket{1}-\ket{3})\otimes(\ket{1}+\ket{3}) \ ,
&&\ket{\psi_3}= \tfrac12 (\ket{1}+i\ket{3})\otimes(\ket{1}+i\ket{3}) \ ,
\nonumber\\[4pt]
&\ket{\psi_4}= \tfrac12 (\ket{1}-i\ket{3})\otimes(\ket{1}-i\ket{3}) \ ,
&&\ket{\psi_5}= \tfrac12 (\ket{1}+\ket{2})\otimes(\ket{1}-\ket{2}) \ ,
&&\ket{\psi_6}= \tfrac12 (\ket{1}-\ket{2})\otimes(\ket{1}+\ket{2}) \ ,
\nonumber\\[4pt]
&\ket{\psi_7}= \tfrac12 (\ket{1}+i\ket{2})\otimes(\ket{1}-i\ket{2}) \ ,
&&\ket{\psi_8}= \tfrac12 (\ket{1}-i\ket{2})\otimes(\ket{1}+i\ket{2}) \ ,
&&\ket{\psi_9}= \ket{1}\otimes\ket{2} \ ,
\nonumber\\[4pt]
&\ket{\psi_{10}}= \ket{2}\otimes\ket{1} \ ,
&&\ket{\psi_{11}}= \ket{1}\otimes\ket{3} \ ,
&&\ket{\psi_{12}}= \ket{3}\otimes\ket{1} \ ,
\end{align}
and weights
\begin{equation}
p_i=\frac{\sqrt{2}-1}{4}\quad (i\le 8) \ ,
\qquad
p_i=\frac{3-2\sqrt{2}}{4}\quad (i\ge 9) \ .
\end{equation}
This proves the claim.
\end{proof}

\subparagraph{Second family.}
\begin{proposition}
\label{prop:rho23k}
The families $\rho_2(k)$ and $\rho_3(k)$ defined in \cref{eq:rho23k} satisfy:
\begin{enumerate}
\item both are PPT for all $0\le k\le \sqrt2$;
\item the fixed witness $C^{(3)}_{1/3,\,2/3}$ detects both families for
\[
\frac54<k\le \sqrt2\,;
\]
\item the CCNR criterion proves entanglement for both families whenever
\[
k>\frac{\sqrt{10}-S}{4}\simeq 1.128  \ ,
\]
where $S$ is the largest real root of $S^3+2S^2-28S-72=0$;
\item the states $\rho_2(1)$ and $\rho_3(1)$ are separable;
\item according to the level-$2$ DPS hierarchy with PPT constraints, both
families are entangled for every $k>1$.
\end{enumerate}
\end{proposition}

\begin{proof}
Consider
\begin{equation}
\rho_2(k)=\frac {1}{10}
\left(
\begin{array}{ccc|ccc|ccc}
2 & \cdot & \cdot & \cdot & \cdot & \cdot & \cdot & \cdot & -k \\
\cdot & 1 & \cdot & -k & \cdot & \cdot & \cdot & \cdot & \cdot \\
\cdot & \cdot & 2 & \cdot & \cdot & \cdot & \cdot & \cdot & \cdot \\ \hline
\cdot & -k & \cdot & 2 & \cdot & \cdot & \cdot & \cdot & \cdot \\
\cdot & \cdot & \cdot & \cdot & 1 & \cdot & \cdot & \cdot & \cdot \\
\cdot & \cdot & \cdot & \cdot & \cdot & \cdot & \cdot & \cdot & \cdot \\ \hline
\cdot & \cdot & \cdot & \cdot & \cdot & \cdot & 1 & \cdot & \cdot \\
\cdot & \cdot & \cdot & \cdot & \cdot & \cdot & \cdot & \cdot & \cdot \\
-k & \cdot & \cdot & \cdot & \cdot & \cdot & \cdot & \cdot & 1
\end{array}
\right),
\qquad
\rho_3(k)=\frac{1}{10}
\left(
\begin{array}{ccc|ccc|ccc}
2 & \cdot & \cdot & \cdot & \cdot & \cdot & \cdot & \cdot & -k \\
\cdot & 2 & \cdot & -k & \cdot & \cdot & \cdot & \cdot & \cdot \\
\cdot & \cdot & 1 & \cdot & \cdot & \cdot & \cdot & \cdot & \cdot \\ \hline
\cdot & -k & \cdot & 1 & \cdot & \cdot & \cdot & \cdot & \cdot \\
\cdot & \cdot & \cdot & \cdot & 1 & \cdot & \cdot & \cdot & \cdot \\
\cdot & \cdot & \cdot & \cdot & \cdot & \cdot & \cdot & \cdot & \cdot \\ \hline
\cdot & \cdot & \cdot & \cdot & \cdot & \cdot & 2 & \cdot & \cdot \\
\cdot & \cdot & \cdot & \cdot & \cdot & \cdot & \cdot & \cdot & \cdot \\
-k & \cdot & \cdot & \cdot & \cdot & \cdot & \cdot & \cdot & 1
\end{array}
\right).
\label{eq:rho23k}
\end{equation}
Both states are PPT for $0\le k\le \sqrt2$.

Evaluating the fixed witness $C^{(3)}_{1/3,\,2/3}$ gives
\[
\Tr\!\bigl[C^{(3)}_{1/3,\,2/3}\,\rho_2(k)\bigr]
=
\Tr\!\bigl[C^{(3)}_{1/3,\,2/3}\,\rho_3(k)\bigr]
=
\frac{5-4k}{15},
\]
so both families are certified as PPT entangled for
$\frac54<k\le \sqrt2$.

The CCNR criterion improves the bound and proves entanglement for
\begin{equation}
k>\frac{\sqrt{10}-S}{4}\simeq 1.128,
\end{equation}
where $S$ is the largest real root of $S^3+2S^2-28S-72=0$, while the
level-$2$ DPS hierarchy with PPT constraints indicates entanglement for every
$k>1$.

At $k=1$ both families are separable: one recovers product
decompositions of the form \cref{eq:rho1-sep}. For $\rho_2(1)$ the
nonzero weights are $p_i=1/10$ for all $i\neq9,12$, while for
$\rho_3(1)$ they are $p_i=1/10$ for all $i\neq10,11$.
\end{proof}

For what follows we introduce the two subspaces
\begin{equation}
\mathcal H_{12}:=\operatorname{span}\{\ket{11},\ket{12},\ket{21},\ket{22}\} \ ,
\qquad
\mathcal H_{13}:=\operatorname{span}\{\ket{11},\ket{13},\ket{31},\ket{33}\} \ .
\label{eq:xi-subspaces}
\end{equation}

\subparagraph{Third family.}
\begin{proposition}
\label{prop:xiak}
Let $a\in[0,1/2]$, and let $\xi(a,k)$ be the family defined in
\cref{eq:xiak}, with coefficients $q,r,s,v,g$ given by \cref{eq:sol}.
Then:
\begin{enumerate}
\item $\xi(a,k)$ is separable for every $0\le k\le v$;
\item the threshold $k=v$ arises from an explicit decomposition into two
$2\times2$ PPT blocks;
\item according to the level-$2$ DPS hierarchy with PPT constraints, the
state is bound entangled for $v<k \leq \sqrt{2}\,v$.
\end{enumerate}
\end{proposition}

\begin{proof}
Consider the family
\begin{equation}
\xi(a,k)=
\left(
\begin{array}{ccc|ccc|ccc}
q & \cdot & \cdot & \cdot & \cdot & \cdot & \cdot & \cdot & -k \\
\cdot & r & \cdot & -k & \cdot & \cdot & \cdot & \cdot & \cdot \\
\cdot & \cdot & s & \cdot & \cdot & \cdot & \cdot & \cdot & \cdot \\ \hline
\cdot & -k & \cdot & s & \cdot & \cdot & \cdot & \cdot & \cdot \\
\cdot & \cdot & \cdot & \cdot & v & \cdot & \cdot & \cdot & \cdot \\
\cdot & \cdot & \cdot & \cdot & \cdot & \cdot & \cdot & \cdot & \cdot \\ \hline
\cdot & \cdot & \cdot & \cdot & \cdot & \cdot & r & \cdot & \cdot \\
\cdot & \cdot & \cdot & \cdot & \cdot & \cdot & \cdot & \cdot & \cdot \\
-k & \cdot & \cdot & \cdot & \cdot & \cdot & \cdot & \cdot & v
\end{array}
\right),
\label{eq:xiak}
\end{equation}
with $0\le k\le g$. Here $q,r,s,v,g$ depend on a single parameter
$a\in[0,1/2]$ through the formulas in \cref{eq:sol}. This family
interpolates among the witness states used in
\cref{sec:decomposability}.

Then
\begin{equation}
\xi(a,k)=\frac12\bigl(\mu (a,k)+\nu(a,k)\bigr) \ ,
\end{equation}
where $\mu (a,k)$ is supported on $\mathcal H_{12}$ and
$\nu(a,k)$ on $\mathcal H_{13}$. In the ordered bases of
\cref{eq:xi-subspaces},
\begin{equation}
\mu(a,k)=2
\begin{pmatrix}
v & 0 & 0 & 0 \\
0 & r & -k & 0 \\
0 & -k & s & 0 \\
0 & 0 & 0 & v
\end{pmatrix},
\qquad
\nu(a,k)=2
\begin{pmatrix}
v & 0 & 0 & -k \\
0 & s & 0 & 0 \\
0 & 0 & r & 0 \\
-k & 0 & 0 & v
\end{pmatrix}.
\label{eq:xi-blocks}
\end{equation}
Since $rs=qv=2v^2$, both blocks are PPT $2\times2$ states whenever
$k\le v$, and are therefore separable by the PPT criterion in
dimension $2\times2$.  Consequently, $\xi(a,k)$ is separable for $0\le
k\le v$.  For $\sqrt{2}\, v \geq k>v$, the level-$2$ DPS hierarchy with PPT constraints
indicates that the state is bound entangled.
\end{proof}

\begin{remark}
The witness states $\rho_1$, $\rho_2$, and $\rho_3$ used in
\cref{sec:decomposability} are recovered from the family
$\rho(a)$ in \cref{eq:bound_entan} at $a=1/3$, $a=1/4$, and $a=2/5$,
respectively. Likewise, $\rho_1 (k)$, $\rho_2(k)$ and $\rho_3(k)$ arise from
$\xi(a,k)$ at those same special values of $a$, up to the normalization of
the deformation parameter.
\end{remark}

\subsection{Families of PPT states efficiently detected by fixed witnesses}

The families $\rho_1(k)$, $\rho_2(k)$, and $\rho_3(k)$ provide natural
witness-adapted deformations, but their PPT-entangled branches are not always
optimally detected by the simplest fixed witnesses. We therefore introduce
alternative deformations whose full PPT-entangled branch is detected by fixed
positive maps.

\begin{proposition}
\label{prop:sigma1}
For $0\le k\le \sqrt2$, define
\[
\sigma_1(k)=
\frac{1}{4(1+k)}
\left(
\begin{array}{ccc|ccc|ccc}
2 & \cdot & \cdot & \cdot & \cdot & \cdot & \cdot & \cdot & -k\\
\cdot & k & \cdot & -k & \cdot & \cdot & \cdot & \cdot & \cdot\\
\cdot & \cdot & k & \cdot & \cdot & \cdot & \cdot & \cdot & \cdot\\ \hline
\cdot & -k & \cdot & k & \cdot & \cdot & \cdot & \cdot & \cdot\\
\cdot & \cdot & \cdot & \cdot & 1 & \cdot & \cdot & \cdot & \cdot\\
\cdot & \cdot & \cdot & \cdot & \cdot & \cdot & \cdot & \cdot & \cdot\\ \hline
\cdot & \cdot & \cdot & \cdot & \cdot & \cdot & k & \cdot & \cdot\\
\cdot & \cdot & \cdot & \cdot & \cdot & \cdot & \cdot & \cdot & \cdot\\
-k & \cdot & \cdot & \cdot & \cdot & \cdot & \cdot & \cdot & 1
\end{array}
\right) \ .
\]
Then $\sigma_1(\sqrt2)=\rho_1$, the state $\sigma_1(k)$ is separable for
$0\le k\le1$ and PPT for $0\le k\le\sqrt2$, and
\[
\Tr\!\bigl[C^{(1)}_{1/2,\,1/2}\,\sigma_1(k)\bigr]
=
\frac{1-k}{2(1+k)} \ .
\]
Hence $\sigma_1(k)$ is PPT entangled and detected by the fixed witness
$C^{(1)}_{1/2,\,1/2}$ for every $1<k\le\sqrt2$.
\end{proposition}

\begin{proof}
The identity $\sigma_1(\sqrt2)=\rho_1$ is immediate from the definitions.
Moreover,
\begin{equation}
\Tr\!\bigl[C^{(1)}_{a,b}\,\sigma_1(k)\bigr]
=
\frac{a+kb-k}{1+k} \ ,
\end{equation}
so the choice $a=b=1/2$ gives the stated expectation value and proves
entanglement for $k>1$.

The stated PPT interval $0\le k\le\sqrt2$ follows directly from the
$2\times2$ principal minors of $\sigma_1(k)$ and its partial
transpose.

For $0\le k\le1$, write
\begin{equation}
\sigma_1(k)=\frac{1}{4(1+k)}\bigl(M_1(k)+N_1(k)\bigr) \ ,
\end{equation}
with
\begin{equation}
M_1(k)=
\begin{pmatrix}
1 & 0 & 0 & 0 \\
0 & k & -k & 0 \\
0 & -k & k & 0 \\
0 & 0 & 0 & 1
\end{pmatrix} \ ,
\qquad
N_1(k)=
\begin{pmatrix}
1 & 0 & 0 & -k \\
0 & k & 0 & 0 \\
0 & 0 & k & 0 \\
-k & 0 & 0 & 1
\end{pmatrix} \ ,
\label{eq:sigma1-blocks}
\end{equation}
written in the ordered bases of \cref{eq:xi-subspaces}. From the the
PPT criterion for $2\times2$ states it follows that $\sigma _1(k)$  is
separable for $0\le k\le 1$.
\end{proof}

The family $\sigma_1(k)$ should be contrasted with the original deformation
$\rho_1(k)$. For $\rho_1(k)$, the simple witness family $C^{(1)}_{a,b}$
detects entanglement only for $k>(1+\sqrt2)/2$, so the interval
$1<k\le(1+\sqrt2)/2$ requires stronger criteria such as CCNR or the DPS
hierarchy. By contrast, the fixed witness $C^{(1)}_{1/2,\,1/2}$ already
detects the entire PPT-entangled branch of $\sigma_1(k)$.

A similar sharpening occurs in the Case~3 geometry. We introduce two explicit
families, $\sigma_2(k)$ and $\sigma_3(k)$, whose endpoints at $k=\sqrt2$
coincide with the distinguished PPT states $\rho_2$ and $\rho_3$ and whose
entire PPT-entangled branch is detected by the fixed witness
$C^{(3)}_{1/3,\,2/3}$.

\begin{proposition}
\label{prop:sigma23}
For $0\le k\le\sqrt2$, define $\sigma_2(k)$ and $\sigma_3(k)$ as follows.
For $0\le k\le1$, set
\[
\sigma_2(k)=\sigma_3(k)=
\frac{1}{8}
\left(
\begin{array}{ccc|ccc|ccc}
2 & \cdot & \cdot & \cdot & \cdot & \cdot & \cdot & \cdot & -k\\
\cdot & 1 & \cdot & -k & \cdot & \cdot & \cdot & \cdot & \cdot\\
\cdot & \cdot & 1 & \cdot & \cdot & \cdot & \cdot & \cdot & \cdot\\ \hline
\cdot & -k & \cdot & 1 & \cdot & \cdot & \cdot & \cdot & \cdot\\
\cdot & \cdot & \cdot & \cdot & 1 & \cdot & \cdot & \cdot & \cdot\\
\cdot & \cdot & \cdot & \cdot & \cdot & \cdot & \cdot & \cdot & \cdot\\ \hline
\cdot & \cdot & \cdot & \cdot & \cdot & \cdot & 1 & \cdot & \cdot\\
\cdot & \cdot & \cdot & \cdot & \cdot & \cdot & \cdot & \cdot & \cdot\\
-k & \cdot & \cdot & \cdot & \cdot & \cdot & \cdot & \cdot & 1
\end{array}
\right) \ .
\]
For $1\le k\le\sqrt2$, define with $j=2,3$
\[
\sigma_j(k)=
\left(
\begin{array}{ccc|ccc|ccc}
q_j(k) & \cdot & \cdot & \cdot & \cdot & \cdot & \cdot & \cdot & -g_j(k)\\
\cdot & r_j(k) & \cdot & -g_j(k) & \cdot & \cdot & \cdot & \cdot & \cdot\\
\cdot & \cdot & s_j(k) & \cdot & \cdot & \cdot & \cdot & \cdot & \cdot\\ \hline
\cdot & -g_j(k) & \cdot & s_j(k) & \cdot & \cdot & \cdot & \cdot & \cdot\\
\cdot & \cdot & \cdot & \cdot & v_j(k) & \cdot & \cdot & \cdot & \cdot\\
\cdot & \cdot & \cdot & \cdot & \cdot & \cdot & \cdot & \cdot & \cdot\\ \hline
\cdot & \cdot & \cdot & \cdot & \cdot & \cdot & r_j(k) & \cdot & \cdot\\
\cdot & \cdot & \cdot & \cdot & \cdot & \cdot & \cdot & \cdot & \cdot\\
-g_j(k) & \cdot & \cdot & \cdot & \cdot & \cdot & \cdot & \cdot & v_j(k)
\end{array}
\right) \ ,
\]
with
\begin{align}
q_2(k)&=\frac{3}{10}+\frac{\sqrt2}{20}-\frac{1+\sqrt2}{20}k \ ,
& r_2(k)&=\frac{3}{20}+\frac{\sqrt2}{40}-\frac{1+\sqrt2}{40}k \ ,
& s_2(k)&=\frac{1}{20}-\frac{3\sqrt2}{40}+\frac{3(1+\sqrt2)}{40}k \ ,
\nonumber\\
v_2(k)&=r_2(k) \ ,
& g_2(k)&=\frac{1}{20}+\frac{\sqrt2}{40}+\frac{3-\sqrt2}{40}k \ .
\end{align}
and 
\begin{align}
q_3(k)&=\frac{3}{10}+\frac{\sqrt2}{20}-\frac{1+\sqrt2}{20}k \ ,
& r_3(k)&=\frac{1}{20}-\frac{3\sqrt2}{40}+\frac{3(1+\sqrt2)}{40}k \ ,
& s_3(k)&=\frac{3}{20}+\frac{\sqrt2}{40}-\frac{1+\sqrt2}{40}k \ ,
\nonumber\\
v_3(k)&=s_3(k)  \ ,
& g_3(k)&=\frac{1}{20}+\frac{\sqrt2}{40}+\frac{3-\sqrt2}{40}k \ .
\end{align}
Then:
\begin{enumerate}
\item $\sigma_2(\sqrt2)=\rho_2$ and $\sigma_3(\sqrt2)=\rho_3$;
\item both families are separable for $0\le k\le1$;
\item both families are PPT for $0\le k\le\sqrt2$;
\item
\[
\Tr\!\bigl[C^{(3)}_{1/3,\,2/3}\,\sigma_2(k)\bigr]
=
\Tr\!\bigl[C^{(3)}_{1/3,\,2/3}\,\sigma_3(k)\bigr]
=
\frac{\sqrt2-3}{15}(k-1) \ ,
\]
so both are PPT entangled and detected by $C^{(3)}_{1/3,\,2/3}$ for every
$1<k\le\sqrt2$.
\end{enumerate}
\end{proposition}

\begin{proof}
For the lower branch $0\le k\le1$, the two families coincide and admit the
block decomposition 
\begin{equation}
\sigma_2(k)=\sigma_3(k)=\frac18\bigl(M_0(k)+N_0(k)\bigr)\ ,
\end{equation}
with
\begin{equation}
M_0(k)=
\begin{pmatrix}
1 & 0 & 0 & 0 \\
0 & 1 & -k & 0 \\
0 & -k & 1 & 0 \\
0 & 0 & 0 & 1
\end{pmatrix} \ ,
\qquad
N_0(k)=
\begin{pmatrix}
1 & 0 & 0 & -k \\
0 & 1 & 0 & 0 \\
0 & 0 & 1 & 0 \\
-k & 0 & 0 & 1
\end{pmatrix} \ ,
\label{eq:sigma23-lower-blocks}
\end{equation}
in the ordered bases of \cref{eq:xi-subspaces}. Hence both families
are separable for $0\le k\le1$ by the PPT criterion for $2\times 2$
systems.

For the upper branch $1\le k\le\sqrt2$, both families admit the common
representation
\begin{equation}
\sigma_j(k)=M_j(k)+N_j(k) \ ,
\qquad j=2,3 \ ,
\end{equation}
where
\begin{equation}
M_j(k)=
\begin{pmatrix}
v_j(k) & 0 & 0 & 0 \\
0 & r_j(k) & -g_j(k) & 0 \\
0 & -g_j(k) & s_j(k) & 0 \\
0 & 0 & 0 & v_j(k)
\end{pmatrix} \ ,
\qquad
N_j(k)=
\begin{pmatrix}
v_j(k) & 0 & 0 & -g_j(k) \\
0 & s_j(k) & 0 & 0 \\
0 & 0 & r_j(k) & 0 \\
-g_j(k) & 0 & 0 & v_j(k)
\end{pmatrix} \ .
\label{eq:sigma23-upper-blocks}
\end{equation}
Since $q_2(k)=2v_2(k)$ and $q_3(k)=2v_3(k)$, these block decompositions
reproduce exactly the matrix forms defining $\sigma_2(k)$ and
$\sigma_3(k)$.

At $k=\sqrt2$, direct substitution gives
\[
\bigl(v_2,r_2,s_2,g_2\bigr)\Big|_{k=\sqrt2}
=
\left(\frac1{10},\frac1{10},\frac15,\frac{\sqrt2}{10}\right) \ ,
\qquad
\bigl(v_3,r_3,s_3,g_3\bigr)\Big|_{k=\sqrt2}
=
\left(\frac1{10},\frac15,\frac1{10},\frac{\sqrt2}{10}\right) \ ,
\]
whence $\sigma_2(\sqrt2)=\rho_2$ and $\sigma_3(\sqrt2)=\rho_3$.

To verify PPT for $1\le k\le\sqrt2$, it is enough to check the two
non-trivial determinants of the partial transpose. For both families one
finds
\begin{equation}
r_j(k)s_j(k)-g_j(k)^2
=
\frac{(k-1)(\sqrt2-k)}{80}\ge 0 \ ,
\qquad j=2,3 \ ,
\end{equation}
and
\begin{equation}
q_j(k)v_j(k)-g_j(k)^2
=
\frac{14+4\sqrt2-(8+6\sqrt2)k+(2\sqrt2-1)k^2}{320}\ge 0 \ ,
\qquad j=2,3 \ ,
\end{equation}
throughout the interval $1\le k\le\sqrt2$. Hence both upper branches are
PPT. Finally, evaluating the fixed witness $C^{(3)}_{1/3,\,2/3}$ yields the
stated expectation value, which is strictly negative for every
$1<k\le\sqrt2$.
\end{proof}

The families $\sigma_2(k)$ and $\sigma_3(k)$ play for $\rho_2$ and
$\rho_3$ the same role that $\sigma_1(k)$ plays for $\rho_1$: they provide
explicit interpolations between a separable regime and a PPT-entangled
regime while preserving the distinguished endpoint states arising from the
decomposability analysis.
\section{Conclusion}
\label{sec:conclusion}

We have presented a complete analytical study of a sparse family of positive
maps on qutrits that originally emerged from an optimization-based search for
non-decomposable maps~\cite{morgillo2026maps}. The main outcome is that this
family provides an explicit low-dimensional laboratory in which complete
positivity, positivity, decomposability, and PPT-entanglement detection can
all be analysed on equal footing.

The key structural ingredient is the sparse Choi-matrix pattern. It reduces
positivity to the analysis of a Hermitian biquadratic form, or equivalently
of a $3\times3$ Hermitian matrix over simplex variables, and thereby leads to
exact positivity regions for three representative parametric families. For the
same families we determined exact decomposability thresholds and constructed
explicit PPT states that witness the non-decomposable sectors. In this sense,
the map-theoretic and state-theoretic sides of the problem remain tightly
linked throughout the analysis.

The state families introduced in \cref{sec:bound_ent_states} sharpen this
picture further. Besides the witness-adapted deformations naturally inherited
from the decomposability analysis, we identified alternative one-parameter
families whose entire PPT-entangled branch is detected by fixed positive
maps. This yields exact thresholds between separability and PPT entanglement
and makes the geometry of the non-decomposable region explicit at the level
of concrete two-qutrit states. For the trace-preserving subclass, we also
compared positivity with a recent eigenvalue bound for $2$-positive
maps~\cite{boundoneigs2025}, thereby making the gap between positivity and
higher-order positivity completely explicit within this family.

Several directions remain open. It would be natural to understand how much of
the present structure survives under perturbations of the sparse Choi pattern,
and whether analogous analytically tractable families can be constructed in
higher dimensions. More broadly, the present example suggests that combining
optimization-based searches with structural analysis may provide an effective
route to further explicit families of non-decomposable maps and associated
PPT-entangled states.
\section{Acknowledgments}
A.R.M. acknowledges support from the PNRR MUR Project PE0000023-NQSTI. This work has been sponsored by PRIN MUR Project 2022SW3RPY.

\bibliographystyle{apsrev4-2}
\bibliography{bibliography5}

\appendix

\section{Kraus--Stinespring expansion}
\label{app:KS}
We derive the Kraus--Stinespring expansion of the map introduced in \cref{eq:map}. According to \cref{KSformula}, once the eigenvectors
\[
\ket{\phi_\alpha}=\sum_{n=1}^3\ket{n}\otimes\ket{\phi_{\alpha n}}
\]
of the Choi matrix $C_\Phi$ in \cref{eq:ch1} are known, the corresponding Kraus operators are obtained from
\begin{equation}
\label{Kraus)p}
F_\alpha^\dagger =
\begin{pmatrix}
\langle 1\vert\phi_{\alpha 1}\rangle & \langle 1\vert\phi_{\alpha 2}\rangle & \langle 1\vert\phi_{\alpha 3}\rangle \\
\langle 2\vert\phi_{\alpha 1}\rangle & \langle 2\vert\phi_{\alpha 2}\rangle & \langle 2\vert\phi_{\alpha 3}\rangle \\
\langle 3\vert\phi_{\alpha 1}\rangle & \langle 3\vert\phi_{\alpha 2}\rangle & \langle 3\vert\phi_{\alpha 3}\rangle
\end{pmatrix}\,.
\end{equation}

The structure of $C_\Phi$ makes its diagonalization essentially explicit: there is one $2\times2$ block associated with the pair $(\ket{11},\ket{33})$, one $2\times2$ block associated with $(\ket{12},\ket{21})$, and five one-dimensional blocks. The eigenvectors listed below reflect precisely this block decomposition.

The Kraus operators in \cref{eq:Kraus1,eq:Kraus2,eq:Kraus3} are obtained from
\begin{eqnarray}
\label{eigv1}
\ket{\phi_{11}}&=&\frac{1}{\sqrt{2}}
\begin{pmatrix}
1\\0\\0
\end{pmatrix}\ ,\ 
\ket{\phi_{12}}=
\begin{pmatrix}
0\\0\\0
\end{pmatrix}\ ,\ 
\ket{\phi_{13}}=\frac{1}{\sqrt{2}}
\begin{pmatrix}
0\\0\\{\rm e}^{-i \arg(z)}
\end{pmatrix}\ ;\\
\label{eigv2}
\ket{\phi_{21}}&=&\frac{1}{\sqrt{2}}
\begin{pmatrix}
1\\0\\0
\end{pmatrix}\ ,\ 
\ket{\phi_{22}}=
\begin{pmatrix}
0\\0\\0
\end{pmatrix}\ ,\ 
\ket{\phi_{23}}=\frac{1}{\sqrt{2}}
\begin{pmatrix}
0\\0\\-{\rm e}^{-i\arg(z)}
\end{pmatrix}\ ;\\
\label{eigv3}
\ket{\phi_{31}}&=&
\begin{pmatrix}
0\\0\\0
\end{pmatrix}\ ,\ 
\ket{\phi_{32}}=
\begin{pmatrix}
0\\0\\0
\end{pmatrix}\ ,\ 
\ket{\phi_{33}}=
\begin{pmatrix}
1\\0\\0
\end{pmatrix}\ ;\\
\label{eigv4}
\ket{\phi_{41}}&=&
\begin{pmatrix}
0\\0\\0
\end{pmatrix}\ ,\ 
\ket{\phi_{42}}=
\begin{pmatrix}
0\\0\\0
\end{pmatrix}\ ,\ 
\ket{\phi_{43}}=
\begin{pmatrix}
0\\1\\0
\end{pmatrix}\ ;\\
\label{eigv5}
\ket{\phi_{51}}&=&
\begin{pmatrix}
0\\0\\0
\end{pmatrix}\ ,\ 
\ket{\phi_{52}}=
\begin{pmatrix}
0\\1\\0
\end{pmatrix}\ ,\ 
\ket{\phi_{53}}=
\begin{pmatrix}
0\\0\\0
\end{pmatrix}\ ;\\
\label{eigv6}
\ket{\phi_{61}}&=&
\begin{pmatrix}
0\\0\\0
\end{pmatrix}\ ,\ 
\ket{\phi_{62}}=
\begin{pmatrix}
0\\0\\1
\end{pmatrix}\ ,\ 
\ket{\phi_{63}}=
\begin{pmatrix}
0\\0\\0
\end{pmatrix}\ ;\\
\label{eigv7}
\ket{\phi_{71}}&=&
\begin{pmatrix}
0\\0\\1
\end{pmatrix}\ ,\ 
\ket{\phi_{72}}=
\begin{pmatrix}
0\\0\\0
\end{pmatrix}\ ,\ 
\ket{\phi_{73}}=
\begin{pmatrix}
0\\0\\0
\end{pmatrix}\ ;\\
\label{eigv8}
\ket{\phi_{81}}&=&\frac{1}{\sqrt{
|w|
^2+(\gamma_8-b)^2}}
\begin{pmatrix}
0\\
|w|
\\0
\end{pmatrix}\ ,\ 
\ket{\phi_{82}}=
\begin{pmatrix}
(\gamma_8-b){\rm e}^{-i\arg(w)} \\0\\0
\end{pmatrix}\ ,\ 
\ket{\phi_{83}}=
\begin{pmatrix}
0\\0\\0
\end{pmatrix}\ ;\\
\label{eigv9}
\ket{\phi_{91}}&=&\frac{1}{\sqrt{
|w|
^2+(\gamma_9-b)^2}}
\begin{pmatrix}
0\\
|w|
\\0
\end{pmatrix}\ ,\ 
\ket{\phi_{92}}=
\begin{pmatrix}
(\gamma_9-b){\rm e}^{-i\arg(w)}\\0\\0
\end{pmatrix}\ ,\ 
\ket{\phi_{93}}=
\begin{pmatrix}
0\\0\\0
\end{pmatrix}\ .
\end{eqnarray}

Substituting these vectors into \cref{Kraus)p} yields the Kraus operators reported in \cref{eq:Kraus1,eq:Kraus2,eq:Kraus3}.

\section{$2$-Positivity}
\label{app3}
We show that, for the family of maps in \cref{eq:map}, $2$-positivity already implies $3$-positivity, and hence complete positivity and decomposability.

It is enough to test positivity of $(\mathrm{id}_2\otimes\Phi)[P]$ on rank-one projectors
$P:=\ket{\Psi}\bra{\Psi}$, where
\begin{equation}
\label{app3.1}
\ket{\Psi}=\ket{1}\otimes\ket{\psi_1}+\ket{2}\otimes\ket{\psi_2}\,,
\qquad
\ket{1}=\begin{pmatrix}1\\ 0\end{pmatrix}\,,
\qquad
\ket{2}=\begin{pmatrix}0\\ 1\end{pmatrix}\,,
\end{equation}
with $\ket{\psi_{1,2}}$ not necessarily normalized vectors in $\mathbb{C}^3$.

Then $2$-positivity requires the positivity of the $6\times 6$ matrix
\begin{equation}
\label{app3.2}
M=\sum_{i,j=1}^2\ket{i}\bra{j}\otimes\Phi[\ket{\psi_i}\bra{\psi_j}]=
\begin{pmatrix}
\Phi[\ket{\psi_1}\bra{\psi_1}]&\Phi[\ket{\psi_1}\bra{\psi_2}]\\
\Phi[\ket{\psi_2}\bra{\psi_1}]&\Phi[\ket{\psi_2}\bra{\psi_2}]
\end{pmatrix}\,.
\end{equation}

By choosing suitable pairs $(\ket{\psi_1},\ket{\psi_2})$, one can
isolate the constraints appearing in \cref{CPcond}. Setting
\[
\ket{\psi_1}=(1,0,0)^T,
\qquad
\ket{\psi_2}=0,
\]
yields the necessary conditions $a,b,c\geq 0$. Choosing instead
\[
\ket{\psi_2}=(0,1,0)^T
\]
singles out the non-trivial block involving $w$ and gives the
condition $bc\geq |w| ^2$.  Finally, taking
\[
\ket{\psi_2}=(0,0,1)^T
\]
isolates the block involving $z$ and yields the remaining condition $
a\geq |z|$.

Therefore, every $2$-positive map in this family necessarily satisfies all the conditions in \cref{CPcond}, and is hence completely positive.

\section{$1$-Positivity}
\label{app4}
We show that the study of positivity for the map in \cref{eq:map} reduces to the analysis of the Hermitian form in \cref{eq:bifo}.

The ($1$-)positivity of $\Phi$ is equivalent to the requirement that
\begin{equation}
\bra{\phi}\Phi[\ket{\chi}\!\bra{\chi}]\ket{\phi}\geq 0\,,
\qquad \forall\, \ket{\phi},\ket{\chi}\in\mathbb{C}^3\,.
\end{equation}
Writing
\[
\ket{\phi}=\sum_{i=1}^3 x_i\ket{i},
\qquad
\ket{\chi}=\sum_{i=1}^3 y_i\ket{i},
\]
with $x_i,y_i\in\mathbb{C}$, one obtains
\begin{equation}
\label{app4.1}
\sum_{i,j;k,\ell=1}^3 x_{\ell} x_k^\ast  y_i y_j^\ast \,\bra{k}\Phi[\ket{i}\!\bra{j}]\ket{\ell}\geq 0\,.
\end{equation}
Thus, the problem reduces to evaluating the action of $\Phi$ on the matrix units $\ket{i}\!\bra{j}$. From \cref{eq:map} one finds
\begin{eqnarray}
&&\Phi[\ket{1}\!\bra{1}]=\begin{pmatrix}
a&0&0\\0&b&0\\0&0&c
 \end{pmatrix}\,,
\qquad
\Phi[\ket{1}\!\bra{2}]=\begin{pmatrix}
0&0&0\\ w&0&0\\0&0&0
 \end{pmatrix}\,,
\qquad
 \Phi[\ket{1}\!\bra{3}]=\begin{pmatrix}
0&0&z\\0&0&0\\0&0&0
 \end{pmatrix}\,,
\\
&&\Phi[\ket{2}\!\bra{1}]=\begin{pmatrix}
0&w^\ast&0\\0&0&0\\0&0&0 \end{pmatrix}\,,
\qquad
\Phi[\ket{2}\!\bra{2}]=\begin{pmatrix}
c&0&0\\ 0&a&0\\0&0&b
 \end{pmatrix}\,,
\qquad
 \Phi[\ket{2}\!\bra{3}]=0\,,
\\
&&\Phi[\ket{3}\!\bra{1}]=\begin{pmatrix}
0&0&0\\0&0&0\\ z^\ast&0&0
 \end{pmatrix}\,,
\qquad
\Phi[\ket{3}\!\bra{2}]=0\,,
\qquad
 \Phi[\ket{3}\!\bra{3}]=\begin{pmatrix}
b&0&0\\0&c&0\\0&0&a
 \end{pmatrix}\,.
\end{eqnarray}
Substituting these expressions into \cref{app4.1} gives the biquadratic form reported in \cref{eq:bifo}.

\section{Case 3 for \texorpdfstring{$1/3<a<1/2$}{1/3<a<1/2}: interior tangency and elimination}
\label{app:caseCinterior}
We derive the interior branch that contributes to the positivity boundary of Case~3. Starting from \eqref{eq:S_case3} and \eqref{eq:du_case3}, and eliminating $w^2$ between the equations
\[
S(u_1,u_2)=0
\qquad\text{and}\qquad
\partial_{u_2}S(u_1,u_2)=0 \ ,
\]
with $u_3=1-u_1-u_2$, one finds the algebraic relation
\begin{equation}
u_2\bigl(a+\beta u_1\bigr)=u_3\bigl(a+\beta u_3\bigr)\,,
\label{eq:app_case3_relation}
\end{equation}
where $\beta :=1-3a$. 
Substituting this relation back into $S=0$ gives an explicit parametrization of the corresponding threshold values of $w$:
\begin{equation}
\label{eq:caseC_app_E}
w^2(u)=\frac{2\,(a+u)\left(1-a+\dfrac{u^2}{a}\right)}{\left(1+\dfrac{u^2}{a(1-3a)}\right)^2}\,,
\end{equation}
with $u:=\beta u_3$. The interior branch of the positivity boundary is
therefore reduced to a one-parameter function, and is obtained by
minimizing $w^2(u)$ over admissible values of $u$.

At an interior minimizer $u=u_\ast$, one has
\[
\frac{d}{du}w^2(u_\ast)=0 \ .
\]
Differentiating \eqref{eq:caseC_app_E} gives
\begin{equation}
\label{eq:caseC_app_dW}
\frac{d}{du}w^2(u)=\frac{2a(3a-1)^2}{\bigl(a(1-3a)+u^2\bigr)^3}\,Q(u)\,,
\end{equation}
where
\begin{equation}
\label{eq:caseC_app_Q}
Q(u)=u^4+2a u^3+6a^2u^2+2a^2(1+a)u-a^2(1-a)(1-3a)\,.
\end{equation}
Thus the stationarity condition $\frac{d}{du}w^2(u)=0$ is equivalent to $Q(u)=0$.

Now let
\begin{equation}
\label{eq:caseC_app_F}
F\bigl(u;w^2\bigr):=a w^2\bigl(a(1-3a)+u^2\bigr)^2-2\,(a+u)\bigl(a(1-a)+u^2\bigr)\,\bigl(a(1-3a)\bigr)^2\,.
\end{equation}
Then \eqref{eq:caseC_app_E} is equivalent to
\[
F(u;w^2)=0 \ .
\]

One can eliminate $u$ from the equations
\[
Q(u) = 0
\qquad\text{and}\qquad
F(u; w^2) = 0
\]
by computing their resultant with respect to $u$. Setting this
resultant, $\mathrm{Res}_u(Q,F)=0$, and removing extraneous factors
yields the polynomial
\begin{equation}
\label{eq:caseC_app_P}
\begin{aligned}
P(a, w^2) = {}&16a(2a - 1)w^8 + 8a(23a^3 - 32a^2 + 12a - 1)w^6 \\
&+ (3a - 1)^2(83a^4 - 102a^3 + 25a^2 + 6a - 1)w^4 \\
&+ 2a(3a - 1)^4(8a^3 - 13a^2 + 3a + 1)w^2 + a(1-a)(3a - 1)^6 \,.
\end{aligned}
\end{equation}
The smallest positive $w$ such that $P(a,w^2)=0$ defines
$w_{\mathrm{int}}(a)$. This is precisely the interior branch that
enters the positivity boundary in Case~3. Accordingly, for $1/3 \le a
\le 1/2$ the positivity border is
\begin{equation}
w_{\max}(a)=\min\left\{a+\sqrt{a-2a^2},\,w_{\mathrm{int}}(a)\right\} \ ,
\end{equation}
as stated in \cref{eq:wmax_case3_full}.

\end{document}